\documentclass[a4paper, 12pt]{article}
\usepackage{amsmath,amssymb,amsfonts,amsthm, mathtools}
\usepackage{fullpage, relsize, url, hyperref}
\usepackage{longtable, threeparttable, booktabs, graphicx, float, rotating, tikz, multirow}
\usepackage{times, txfonts, color}
\usepackage{titlesec}
\usepackage{algorithmic}
\usepackage[linesnumbered,ruled]{algorithm2e}
\usepackage{subfigure}
\usepackage{mathtools}
\usepackage{mathrsfs}
\usepackage{natbib} 
\usepackage{array}
\usepackage{arydshln}
\newtheorem{theorem}{Theorem}

\newtheorem{prop}[theorem]{Proposition}

\theoremstyle{definition}
\newtheorem{definition}{Definition}
\newtheorem{remark}{Remark}

\protect

\DeclareMathOperator*{\argmax}{arg\,max}
\DeclareMathOperator*{\argmin}{arg\,min}
\newcommand{\minop}{\operatorname{\triangle}}
\newcommand{\p}{\operatorname{P}\,}
\newcommand{\E}{\operatorname{E}\,}
\newcommand{\Var}{\operatorname{Var}\,}

\begin{document}

 \title{\bf Modified Galton-Watson processes with immigration under an alternative offspring mechanism}
 \author{Wagner Barreto-Souza\footnote{Email: \href{mailto:wagner.barretosouza@kaust.edu.sa}{wagner.barretosouza@kaust.edu.sa} (Corresponding Author)},\,\, Sokol Ndreca$^\star$\footnote{Email: \href{mailto:sokol@est.ufmg.br}{sokol@est.ufmg.br}},\,\, 
 Rodrigo B. Silva$^\P$\footnote{Email: \href{mailto:rodrigo@de.ufpb.br}{rodrigo@de.ufpb.br}}\,\, and 
 Roger W.C. Silva$^\star$\footnote{Email: \href{mailto:rogerwcs@est.ufmg.br}{rogerwcs@est.ufmg.br}} \\\\
 	\small $^*$\it Statistics Program, King Abdullah University of Science and Technology, Thuwal, Saudi Arabia\\
 	\small $^\star$\it Departamento de Estat\' \i stica, Universidade Federal de Minas Gerais, Belo Horizonte, Brazil\\
 	\small $^\P$\it Departamento de Estat\' \i stica, Universidade Federal da Para\'iba, Jo\~ao Pessoa, Brazil}
 \maketitle

\begin{abstract}
We propose a novel class of count time series models alternative to the classic Galton-Watson process with immigration (GWI) and Bernoulli offspring. A new offspring mechanism is developed and its properties are explored. This novel mechanism, called geometric thinning operator, is used to define a class of modified GWI (MGWI) processes, which induces a certain non-linearity to the models. We show that this non-linearity can produce better results in terms of prediction when compared to the linear case commonly considered in the literature. We explore both stationary and non-stationary versions of our MGWI processes. Inference on the model parameters is addressed and the finite-sample behavior of the estimators investigated through Monte Carlo simulations. Two real data sets are analyzed to illustrate the stationary and non-stationary cases and the gain of the non-linearity induced for our method over the existing linear methods. A generalization of the geometric thinning operator and an associated MGWI process are also proposed and motivated for dealing with zero-inflated or zero-deflated count time series data.
\end{abstract}



\section{Introduction}\label{sec:intro}

The Galton-Watson (GW) or branching process is a simple and well-used model for describing populations evolving in time. It is defined by a sequence of non-negative integer-valued random variables $\{X_t\}_{t\in\mathbb{N}}$ satisfying
\begin{eqnarray}\label{gw}
X_t=\sum_{k=1}^{X_{t-1}}\zeta_{t,k},	\quad t\in\mathbb N,
\end{eqnarray}	
with $X_0=1$ by convention, where $\{\zeta_{t,k}\}_{t,k\in\mathbb{N}}$ is a doubly infinite array of independent and identically distributed (iid) random variables. Write $F$ for the offspring distribution, so that $\zeta_{t,k}\sim F$ for all $t,k\geq 1$.  In the populational context, the random variable $X_t$ denotes the size of the $t$-th generation. A generalization of this model is obtained by allowing an independent immigration process in (\ref{gw}), which is known as the GW process with immigration (GWI) and given by
\begin{eqnarray}\label{gwi}
X_t=\sum_{k=1}^{X_{t-1}}\zeta_{t,k}+\epsilon_t,\quad t\in\mathbb N,	
\end{eqnarray}	
where $\{\epsilon_t\}_{t\in\mathbb{N}}$ is assumed to be a sequence of iid non-negative integer-valued random variables, with $\epsilon_t$ independent of $X_{s-1}$ and $\zeta_{s,k}$, for all $k\geq1$ and for all $s\leq t$. If one assumes $\alpha\equiv E(\zeta_{s,k})<\infty$ and $\mu_\epsilon\equiv E(\epsilon_t)<\infty$, then the conditional expectation of the size of the $t$-th generation given the size of the $(t-1)$-th generation, is linear on $X_{t-1}$ and given by  
\begin{eqnarray}\label{condexpectgwi}
E(X_t|X_{t-1})=\alpha X_{t-1}+\mu_\epsilon.
\end{eqnarray}		
	
An interesting example appears when the offspring is Bernoulli distributed, that is, when $\p(\zeta=1)=1-\p(\zeta=0)=\alpha\in(0,1)$. This yields the binomial thinning operator ``$\circ$" by \cite{steandvan1979}, which is defined by $$\alpha\circ X_{t-1}\equiv\sum_{k=1}^{X_{t-1}}\zeta_{t,k}.$$ In this case, the GWI process in (\ref{gwi}) is related to the first-order Integer-valued AutoRegressive (INAR) models presented in \cite{alzalo1987}, \cite{mck1988}, and \cite{dioetal1995}. Conditional least squares estimation for the GWI/INAR models were explored, for instance, by \cite{weiwin1990}, \cite{ispetal2003}, \cite{fremcc2005}, and \cite{rah2008}. 

Alternative integer-valued processes based on non-additive innovation through maximum and minimum operations were proposed by \cite{lit1992}, \cite{lit1996}, \cite{kal1995}, \cite{scoetal2016}, and \cite{aleris2021}. For the count processes $\{X_t\}_{t\in\mathbb N}$ considered in these works, a certain non-linearity is induced in the sense that the conditional expectation $E(X_t|X_{t-1})$ is non-linear on $X_{t-1}$  (and also the conditional variance) in contrast with (\ref{condexpectgwi}). We refer to these models as ``non-linear" along with this paper. On the other hand, the immigration interpretation in a populational context is lost due to the non-additive innovation assumption.

Our aim in this paper is to introduce an alternative model to the classic GW process with immigration (GWI) and Bernoulli offspring. We develop a modified GWI process (MGWI) based on a new thinning operator/offspring mechanism while preserving the additive innovation, which has a practical interpretation. We show that this new mechanism, called the geometric thinning operator, induces a certain non-linearity when compared to the classic GWI/INAR processes. We now highlight other contributions of the present paper: 
\begin{itemize}
\item[(i)] development of inferential procedures and numerical experiments, which are not well-explored for the existing non-linear models aforementioned;
\item[(ii)] properties of the novel geometric thinning operator are established; 
\item[(iii)] a particular MGWI process with geometric marginals is investigated in detail, including an explicit expression for the autocorrelation function; 
\item[(iv)] both stationary and non-stationary cases are explored, being the last important for allowing the inclusion of covariates, a feature not considered by the current non-linear models; 
\item[(v)] empirical evidences that the non-linearity induced for our MGWI processes can produce better results in terms of prediction when compared to the linear case (commonly considered in the literature); 
\item[(vi)] a generalization of the geometric thinning operator and an associated MGWI process are also proposed and motivated for dealing with zero-inflated or zero-deflated count time series data.
\end{itemize}

The paper is organized as follows. In Section \ref{sec:novel_operator}, we introduce the new geometric thinning operator and explore its properties. Section \ref{sec:gwi_process} is devoted to the development of the modified Galton-Watson processes with immigration based on the new operator, with a focus on the case where the marginals are geometrically distributed. Two methods for estimating the model parameters are discussed in Section \ref{sec:inference}, including Monte Carlo simulations to evaluate the proposed estimators. In Section \ref{sec:nonstat}, we introduce a non-stationary MGWI process allowing for the inclusion of covariates and provide some Monte Carlo studies. Section \ref{sec:application} is devoted to two real data applications. Finally, in Section \ref{sec:generalization}, we develop a generalization of the geometric thinning operator and an associated modified GWI model.

\section{Geometric thinning operator: definition and properties}\label{sec:novel_operator}

In this section, we introduce a new thinning operator and derive its main properties. We begin by introducing some notation. For two random variables $X$ and $Y$, we write $\min\{X, Y\}: =X \wedge Y$ to denote the minimum between $X$ and $Y$. The probability generating function (pgf) of a non-negative integer-valued random variable $Y$ is denoted by
\begin{eqnarray*}
\Psi_{Y}(s)= \operatorname{E}\left(s^{Y}\right)=\sum_{k=0}^\infty s^k \p(Y=k), 
\end{eqnarray*}
for all values of $s$ for which the right-hand side converges absolutely. The $n$-th derivative of $\Psi_{Y}(x)$ with respect to $x$ and evaluated at $x=x_0$ is denoted by $\Psi_{Y}^{(n)}(x_0)$. 

Let $Z$ be a geometric random variable with parameter $\alpha>0$ and probability function assuming the form 
$$\p(Z = k) = \frac{\alpha^k}{(1+\alpha)^{k+1}},\quad k=0,1,\dots. $$

In this case, the pgf of $X$ is
\begin{eqnarray}\label{X_pgf}
\Psi_{Z}(s) = \frac{1}{1+\alpha(1-s)},\quad |s|<1+\alpha^{-1},
\end{eqnarray}
and the parameter $\alpha$ has the interpretation $\alpha=E(Z)>0$. The shorthand notation $Z\sim  \mathrm{Geo}(\alpha)$ will be used throughout the text.
We are ready to introduce the new operator and explore some of its  properties.

\begin{definition}(Geometric thinning operator) Let $X$ be a non-negative integer-valued random variable, independent of $Z\sim  \mathrm{Geo}(\alpha)$, with $\alpha>0$. The geometric thinning operator $\triangle$ is defined by
\begin{equation}\label{minop}
\alpha \minop X \equiv \min\left(X, Z\right).
\end{equation}
\end{definition}

\begin{remark}
The operator $\minop$ defined in (\ref{minop}) satisfies $\alpha \minop X \leq X$, like the classic binomial thinning operator $\circ$. Therefore, $\minop$ is indeed a thinning operator.
\end{remark}

In what follows, we present some properties of the proposed geometric thinning operator. We start by obtaining its probability generating function.

\begin{prop}\label{minpgf}
Let $X$ be a non-negative integer-valued random variable with pgf $\Psi_X$. Then, the pgf of $\alpha \minop X$ is given by 
\begin{eqnarray*}
\Psi_{\alpha \minop X}(s)=\frac{1 +\alpha(1-s) \Psi_{X}\left(\dfrac{\alpha s}{1+\alpha}\right)}{1+\alpha(1 - s)},\quad |s|<1+\alpha^{-1}.
\end{eqnarray*}
\end{prop}

\begin{proof} 
By the independence assumption between $X$ and $Z$, it holds that
\begin{align*}
\p(\alpha \minop X = k) & = \p(\alpha \minop X \ge k) - \p(\alpha \minop X \ge k+1)\\&= \p(  Z \ge k) \p( X \ge k) - \p( Z \ge k+1) \p( X \ge k+1)\\
& = \left(\frac{\alpha}{1+\alpha}\right)^k\left[ \p( X = k) +  \frac{1}{1+\alpha}  \p( X \ge k+1)\right].
\end{align*}

Hence,
\begin{align*}
\Psi_{\alpha \minop X}(s)
&=\sum_{k=0}^\infty \left(\frac{\alpha s}{1+\alpha}\right)^k \p( X = k) + 
\frac{1}{1+\alpha} \sum_{k=0}^\infty \left(\frac{\alpha s}{1+\alpha}\right)^k  \p( X \ge k+1)\\
&=\Psi_{X}\left(\frac{\alpha s}{1+\alpha}\right) - \frac{1}{1+\alpha} \Psi_{X}\left(\frac{\alpha s}{1+\alpha}\right)
+ \frac{1}{1+\alpha} \sum_{k=0}^\infty \left(\frac{\alpha s}{1+\alpha}\right)^k  \p( X \ge k)\\
&= \frac{\alpha}{1+\alpha}\Psi_{X}\left(\frac{\alpha s}{1+\alpha}\right) +
\frac{1}{1+\alpha} \sum_{k=0}^\infty \left(\frac{\alpha s}{1+\alpha}\right)^k  \p( X \ge k).  
\end{align*}

The second term on the last equality can be expressed as
\begin{align*}
\frac{1}{1+\alpha} \sum_{k=0}^\infty \left(\frac{\alpha s}{1+\alpha}\right)^k  \p( X \ge k) &=
\frac{1}{1+\alpha} \sum_{k=0}^\infty \left(\frac{\alpha s}{1+\alpha}\right)^k \sum_{l=k}^\infty \p( X = l)\\&= \frac{1}{1+\alpha} \sum_{l=0}^\infty \sum_{k=0}^l \left(\frac{\alpha s}{1+\alpha}\right)^k \p( X = l)\\
&=  \frac{1}{1+\alpha -\alpha s} \left[1- \frac{\alpha s}{1+\alpha} \Psi_{X}\left(\frac{\alpha s}{1+\alpha}\right) \right].
\end{align*}

The result follows by rearranging the terms.
\end{proof}

The next result gives us the moments of $\alpha\minop X$, which will be important to discuss prediction and forecasting in what follows.

\begin{prop}\label{op_mean} Let $\triangle$ be the geometric thinning operator in \eqref{minop}. It holds that
the $n$-th factorial moment of $\alpha \minop X$ is given by
$$
\E\big((\alpha \minop X)_n\big)=n!\alpha^n\left\{1-\sum_{k=0}^{n-1}\dfrac{\Psi_{X}^{(k)}\left(\frac{\alpha }{1+\alpha}\right)}{k!(1+\alpha)^k}\right\},
$$ 
for $n\in\mathbb N$, where $(\alpha \minop X)_n\equiv \alpha \minop X\times(\alpha \minop X-1)\times\ldots\times (\alpha \minop X-n+1)$.
\end{prop}

\begin{proof} 
The result follows by using the pgf given in Proposition \ref{minpgf} and the generalized Leibniz rule for derivatives, namely $(d_1d_2)^{(n)}(s)=\sum_{k=0}^n \binom{n}{k} d_1^{(n-k)}(s)d_2^{(k)}(s)$, with $d_1(s)=1 +\alpha(1-s) \Psi_{X}\left(\dfrac{\alpha s}{1+\alpha}\right)$ and  $d_2(s)=\dfrac{1}{1+\alpha(1 - s)}$.
\end{proof}

In what follows, the notation $X\Rightarrow Y$ means $X$ weakly converges to $Y$.
\begin{prop}
Let $\triangle$ be the geometric thinning operator in \eqref{minop}. Then,
\begin{itemize}
\item[ (i) ]
$
 \alpha \minop X \Rightarrow 0, \quad \text{as } \alpha\to0,
$
\item[ (ii) ]
$
 \alpha \minop X \Rightarrow X, \quad \text{as } \alpha\to\infty.
$
\end{itemize}
\end{prop}

\begin{proof}
The proof follows immediately from Proposition \ref{minpgf} and the Continuity Theorem for pgf's. 
\end{proof}

We now show a property of the operator $\minop$ of own interest.

\begin{prop}
Let $Z_{1},\ldots,Z_{n}$ be independent geometric random variables with parameters $\alpha_1,\ldots,\alpha_n$, respectively. Assume that $X_1, \ldots, X_n$  are non-negative integer-valued random variables independent of the $Z$'s, and let $\alpha_i\minop X_i=\min\left(X_i,Z_i\right)$. Then,
\begin{eqnarray}\label{minalpha}
\wedge_{k=1}^n \alpha_k \minop X_k = \widetilde{\alpha}_n  \minop \wedge_{k=1}^n X_k,
\end{eqnarray}
with $\widetilde{\alpha}_n=\dfrac{\prod_{k=1}^n \alpha_k}{\prod_{k=1}^n(1+ \alpha_k) - \prod_{k=1}^n \alpha_k}$,  $n\in\mathbb N$.
\end{prop}
\begin{proof}
We prove \eqref{minalpha} by induction on $n$. For $n=2$, it holds that
$$
\wedge_{k=1}^2 \alpha_k \minop X_k = \wedge_{k=1}^2 (X_k  \wedge Z_{k})=
(X_1\wedge X_2)\wedge(Z_{1}\wedge Z_{2})=
\widetilde{\alpha}_2 \minop \wedge_{k=1}^2 X_k,
$$
where $\widetilde{\alpha}_2=\dfrac{\prod_{k=1}^2 \alpha_k}{\prod_{k=1}^2(1+ \alpha_k) - \prod_{k=1}^2 \alpha_k}$. Assume that 
$\wedge_{k=1}^{n-1} \alpha_k \minop X_k = \widetilde{\alpha}_{n-1}  \minop \wedge_{k=1}^{n-1} X_k$. 
Since 
$$\wedge_{k=1}^n \alpha_k \minop X_k =(\wedge_{k=1}^{n-1} \alpha_k \minop X_k) \wedge (\alpha_n \minop X_n) =
(\,\widetilde{\alpha}_{n-1}  \minop \wedge_{k=1}^{n-1} X_k) \wedge (\alpha_n \minop X_n) = 
 \widetilde{\alpha}_n  \minop \wedge_{k=1}^n X_k,$$ the proof is complete.
 \end{proof}

In the next section, we introduce our class of modified GW processes and provide some of their properties.

\section{Modified Galton-Watson processes with immigration}\label{sec:gwi_process}

In this section, we introduce modified GW processes with immigration based on the new geometric thinning operator $ \minop$ defined in Section \ref{sec:novel_operator} and explore a special case when the marginals are geometrically distributed.

\begin{definition}
A sequence of random variables $\{X_t\}_{t\in\mathbb{N}}$ is said to be a modified GW process with immigration (in short MGWI) if it satisfies the stochastic equation
\begin{equation}\label{mgwi}
X_t = \alpha \minop X_{t-1} + \epsilon_t,\quad t\in\mathbb{N},
\end{equation}
with $\alpha \minop X_{t-1} =\min\left( X_{t-1},Z_{t}\right)$, $\left\lbrace Z_{t}\right\rbrace_{t\in\mathbb{N}}$ being a sequence of iid random variables with $Z_{1}{\sim}\mbox{Geo}(\alpha)$, $\{\epsilon_t\}_{t\in\mathbb{N}}$ being an iid non-negative integer-valued random variables called innovations, where $\epsilon_t$ is independent of $X_{t-l}$ and $Z_{t-l+1}$, for all $l\geq 1$, and $X_0$ is some starting value/random variable.
\end{definition}

The following theorem is an important result concerning the modified GWI process. 

\begin{theorem} 
	The stochastic process $\{X_t\}_{t\in\mathbb{N}}$ in (\ref{mgwi}) is stationary and ergodic.
\end{theorem}

\begin{proof} Consider the process $\{W_t\}_{t\in\mathbb N}=\{(Z_{t},\epsilon_t)\}_{t\in\mathbb N}$. Since this is a sequence of iid bivariate vectors, it follows that $\{W_t\}_{t\geq 1}$ is stationary and ergodic. Now, note that there is a real function $\xi$, which does not depend on $t$, such that $$X_t=\xi\left((Z_{s},\epsilon_s),\,1\leq s\leq t\right).$$ Hence, the result follows by applying Theorem 36.4 from \cite{bill1995}.  
\end{proof}

From now on, we focus our attention on a special case from our class of MGWI processes when the marginals are geometrically distributed. To do this, let us first discuss the zero-modified geometric (ZMG) distribution, which will play an important role in our model construction. 

We say that a random variable $Y$ follows a ZMG distribution with parameters $\mu>0$ and $\pi\in(-1/\mu,1)$ if its probability function is given by
\begin{eqnarray*}
\p(Y = k) =
\left\{\begin{array}{ll}
	 \pi+(1-\pi)\dfrac{1}{(1+\mu)}, & \text{for } k=0,\\
	 (1-\pi)\dfrac{\mu^k}{(1+\mu)^{k+1}}, & \text{for } k=1,2,\ldots.
\end{array}\right.
\end{eqnarray*}
We denote $Y\sim\mathrm{ZMG}(\pi,\mu)$. The geometric distribution with mean $\mu$ is obtained as a particular case when $\pi=0$. For $\pi<0$ and $\pi>0$, the ZMG distribution is zero-deflated or zero-inflated with relation to the geometric distribution, respectively. The associated pgf assumes the form
\begin{eqnarray}\label{ZMG_pgf}
\Psi_{Y}(s) = \frac{1+\pi\mu(1-s)}{1+\mu(1-s)},\quad |s|<1+\mu^{-1}.
\end{eqnarray}

Now, assume that $X\sim\mathrm{Geo}(\mu)$, with $\mu>0$. We have that
\begin{align*}
\p(\alpha \minop X > z) &= \p(X > z)\p(Z > z) = \left[\left(\frac{\mu}{1+\mu}\right)\left(\frac{\alpha}{1+\alpha}\right)\right]^{z+1},  \quad z=0,1,\dots,
\end{align*}
which means $\alpha \minop X\sim\mathrm{Geo}\left(\dfrac{\alpha\mu}{1+\alpha + \mu}\right)$. From (\ref{mgwi}), it follows that a MGWI process with geometric marginals is well-defined if the function $\Psi_{\epsilon_1}(s)\equiv \dfrac{\Psi_X(s)}{\Psi_{\alpha \minop X}(s)}$ is a proper pgf, with $s$ belonging to some interval containing the value 1, where $\Psi_X(s)$ and $\Psi_{\alpha \minop X}(s)$ are the pgf's of geometric distributions with means $\mu$ and $\dfrac{\alpha\mu}{1+\alpha + \mu}$, respectively.
More specifically, we have 
\begin{equation}\label{epsilon_pgf}
\Psi_{\epsilon_1}(s) = \dfrac{1+\frac{\alpha}{1+\mu+\alpha}\mu(1-s)}{1+\mu(1-s)},\quad |s|<1+\mu^{-1},
\end{equation}
which corresponds to the pgf of a zero-modified geometric distribution with parameters $\mu$ and $\alpha/(1+\mu+\alpha)$; see (\ref{ZMG_pgf}). This enables us to define a new geometric process as follows.

\begin{definition}\label{D:def_mgwi}
The stationary geometric MGWI (Geo-MGWI) process $\{X_t\}_{t\in\mathbb{N}}$ is defined by assuming that (\ref{mgwi}) holds with $\{\epsilon_t\}_{t\in\mathbb{N}}\stackrel{iid}{\sim}\mathrm{ZMG}\left(\dfrac{\alpha}{1+\mu+\alpha},\mu\right)$ and $X_0\sim\mathrm{Geo}(\mu)$.
\end{definition}

From (\ref{epsilon_pgf}), we have that the mean and variance of the innovations $\{\epsilon_t\}_{t\geq1}$ are given by
\begin{equation*}
\mu_\epsilon := \E(\epsilon_t) = \frac{\mu(1+\mu)}{1+\mu+\alpha} \quad \text{and} \quad \sigma_\epsilon^2 := \Var(\epsilon_t) = \frac{\mu(1+\mu)}{1+\mu+\alpha}\left[1+\frac{\mu(1+\mu+2\alpha)}{1+\mu+\alpha}\right],
\end{equation*}
respectively. Additionally, the third and forth moments of the innovations are
\begin{eqnarray*}
E(\epsilon_t^3)=\dfrac{\mu(1+\mu)}{1+\mu+\alpha}(6\mu^2+4\mu+1)\quad\mbox{and}\quad E(\epsilon_t^4)=\dfrac{\mu(1+\mu)}{1+\mu+\alpha}(24\mu^3+36\mu^2+12\mu+5).
\end{eqnarray*}

In what follows, we assume that $\{X_t\}_{t\in\mathbb{N}}$ is a Geo-MGWI process and explore some of its properties. We start with the 1-step transition probabilities.

\begin{prop}\label{P:trans_prob}
The 1-step transition probabilities of the MGWI process, say $\p(x,y) \equiv \p(X_t=y \,|\,X_{t-1}=x)$, assumes the form
\begin{equation}\label{E:trans_prob}
    \p(x,y) =
    \begin{cases}
      \displaystyle\sum_{k=0}^{x-1}\p(Z = k)\p(\epsilon_t = y-k) + \p(Z \geq x)\p(\epsilon_t = y-x), & \text{for}\ \, x \leq y,\\
      \displaystyle\sum_{k=0}^{y}\p(Z = k)\p(\epsilon_t = y-k), & \text{for}\ \, x > y, \\
    \end{cases}
\end{equation}
for $x,y =0,1,\dots$. In particular, we have $\p(0,y) = \p(\epsilon_t = y)$.
\end{prop}
\begin{proof}
For $x=0$, we have that $\p(0,y) = \p(\alpha \minop 0 + \epsilon_t = y) = \p(\epsilon_t = y)$. For $x > 0$, it follows that $$\p(x,y) = \p(\alpha \minop x + \epsilon = y) = \sum_{k=0}^y \p(\alpha \minop x = k)\p(\epsilon = y-k),$$
where 
\begin{equation*}
    \p(\alpha \minop x = z) = 
    \begin{cases}
       0, & \text{for}\ \,x<z, \\
      \p(Z \geq z), & \text{for}\ \,x=z,\\
      \p(Z=z), & \text{for}\ \,x>z. \\
    \end{cases}
  \end{equation*}
This gives the desired transition probabilities in~\eqref{E:trans_prob}.
\end{proof}

\begin{prop}
The joint pgf of the discrete random vector $(X_t, X_{t-1})$ is given by
\begin{equation}
\Psi_{X_t, X_{t-1}}(s_1, s_2) = \frac{\Psi_{\epsilon}(s_1)}{1-\alpha(s_1-1)}\left[\Psi_X(s_2)-\alpha(s_1-1)\Psi_X\left(\frac{s_1s_2\alpha}{1+\alpha}\right)\right], 
\end{equation}
with $\Psi_X(\cdot)$ and $\Psi_\epsilon(\cdot)$ as in \eqref{X_pgf} and \eqref{epsilon_pgf}, respectively, where $s_1$ and $s_2$ belong to some intervals containing the value 1.
\end{prop}
\begin{proof}
We have that 
$$
\Psi_{X_t, X_{t-1}}(s_1, s_2) = \E\left(s_1^{X_t}s_2^{X_{t-1}}\right) = \E\left(s_1^{\alpha\minop X_{t-1}+ \epsilon_t}s_2^{X_{t-1}}\right)  = \Psi_{\epsilon_t}(s_1)\E\left(s_2^{X_{t-1}}\E\left(s_1^{\alpha\minop X_{t-1}}\,|\,X_{t-1}\right)\right),
$$
where
\begin{align}\label{gen_1}
\operatorname{E}\left(s_1^{\alpha \minop X}\,|\,X=x\right) &= \sum_{k=0}^{x-1}s_1^k\p(Z=k) + s_1^x \p(Z\geq x) = \frac{1-\alpha(s_1-1)\left[s_1\alpha/(1+\alpha)\right]^x}{1-\alpha(s_1-1)}.
\end{align}
Therefore,
\begin{align*}
\Psi_{X_t, X_{t-1}}(s_1, s_2) &= \Psi_{\epsilon_t}(s_1)\E\left(\frac{s_2^X}{1-\alpha(s_1-1)}-\frac{\alpha(s_1-1)}{1-\alpha(s_1-1)}\left(\frac{s_1s_2\alpha}{1+\alpha}\right)^X\right)\\
&= \frac{\Psi_{\epsilon_t}(s_1)}{1-\alpha(s_1-1)}\left[\Psi_X(s_2) - \alpha(s_1-1)\Psi_X\left(\frac{s_1s_2\alpha}{1+\alpha}\right)\right].
\end{align*}
\end{proof}

\begin{prop}\label{1_lag}
The 1-step ahead conditional mean and conditional variance are given by
\begin{align*}
\E(X_t\,|\,X_{t-1}) &= \alpha \left[1 - \left(\displaystyle\frac{\alpha}{1+\alpha}\right)^{X_{t-1}}\right] + \mu_\epsilon,\\
\Var(X_t\,|\,X_{t-1}) &=\alpha \left[1-\left(\frac{\alpha}{1+\alpha}\right)^{X_{t-1}}\right]
\left[1+\alpha\left(1+\left(\frac{\alpha}{1+\alpha}\right)^{X_{t-1}}\right)\right]-2\alpha X_{t-1}\left(\frac{\alpha}{1+\alpha}\right)^{X_{t-1}} + \sigma_\epsilon^2,
\end{align*}
respectively.
\end{prop}

\begin{proof}
From the definition of the MGWI process, we obtain that
\begin{equation*}
\E(X_t\,|\,X_{t-1}=x) = \E(\alpha \minop X_{t-1} + \epsilon_t \,|\,X_{t-1}=x)
=\E(\alpha \minop X_{t-1}\,|\,X_{t-1}=x)+\mu_{\epsilon},
\end{equation*}
for all $x=0,1,\dots$. The conditional expectation above can be obtained from Proposition \ref{op_mean} with $X$ being a degenerate random variable at $x$ (i.e. $P(X=x)=1$). Then, it follows that
\begin{equation*}
\E(X_t\,|\,X_{t-1}=x)=\alpha \left[1-\left(\frac{\alpha}{1+\alpha}\right)^{x}\right] + \mu_\epsilon.
\end{equation*}

The conditional variance can be derived analogously, so details are omitted.
\end{proof}

\begin{remark}
Note that the conditional expectation and variance given in Proposition \ref{1_lag} are non-linear on $X_{t-1}$ in contrast with the classic GW/INAR processes where they are linear.
\end{remark}

\begin{prop}
The autocovariance and autocorrelation functions at lag 1 of the Geo-MGWI process are respectively given by
\begin{equation}
\gamma(1) \equiv \operatorname{Cov}(X_t, X_{t-1}) = \frac{\mu\alpha(1+\mu)(1+\alpha)}{(1+\mu+\alpha)^2}\quad \mbox{and}\quad \rho(1) \equiv \operatorname{Corr}\,(X_t, X_{t-1}) = \frac{\alpha(1+\alpha)}{(1+\mu+\alpha)^2}.
\end{equation}
\end{prop}

\begin{proof}
We have that $\operatorname{Cov}(X_t, X_{t-1}) = \E(X_tX_{t-1}) - \E(X_t)\E(X_{t-1})$, with
\begin{align*}
\E(X_tX_{t-1}) &= \E\left[\E(X_tX_{t-1}\,|\,X_{t-1})\right] = \E\left[X_{t-1}\E(X_t\,|\,X_{t-1})\right]\\
&= \alpha\E(X_{t-1}) - \alpha \E\left[X_{t-1}\left(\frac{\alpha}{1+\alpha}\right)^{X_{t-1}}\right] + \mu_\epsilon \E(X_{t-1})\\
&= \mu \alpha - \frac{\mu\alpha^2(1+\alpha)}{(1+\mu+\alpha)^2} + \frac{\mu^2(1+\mu)}{1+\mu+\alpha}.
\end{align*}
After some algebra, the result follows.
\end{proof}

In the following proposition, we obtain an expression for the conditional expectation $E(X_t|X_{t-k}=\ell)$. This function will be important to find the autocovariance function at lag $k\in\mathbb N$ and to perform prediction and/or forecasting.

\begin{prop}\label{cov_lag_k} For $\alpha>0$, define $h_j=\frac{(1+\alpha)^{j-1}}{(1+\alpha)^j-\alpha^j}$ and $g_j=\frac{\alpha(1+\alpha)^{j-1}-\alpha^j}{(1+\alpha)^j-\alpha^j}$, and the real functions
\begin{equation*}\label{coup}
 f_j(x)=
\Psi_{\epsilon_1}({\alpha_*}^{j-1})\left(h_j+g_jx\right), 
\end{equation*}
$j=2,3,\dots$, where $\alpha_*\equiv\frac{\alpha}{1+\alpha}$ and $x\in\mathbb{R}$. 
Finally, let $H_{k}(x)=f_2(\dots  (f_{k-1}(f_k(x))))$. Then, for all $\ell\in\mathbb N_*\equiv\mathbb N\cup\{0\}$, 
\begin{equation}\label{conc_p9}
E(X_t|X_{t-k}=\ell)=\alpha\left(1-H_k\left({\alpha_*}^{k\ell}\right)\right)+\mu_{\epsilon},
\end{equation} 
for all integer $k\geq 2$.
\end{prop}

\begin{proof} Let $\mathcal{F}_t=\sigma(X_1,\dots,X_t)$ denote the sigma-field generated by the random variables $X_1,\dots,X_t$. By the Markov property it is clear that 
\begin{equation}\label{tower}E(X_t|X_{t-k})=E(X_t|\mathcal{F}_{t-k})=E[E(X_t|\mathcal{F}_{t-k+1})|\mathcal{F}_{t-k}],\end{equation} 
for all $k\geq 1$. The proof proceeds by induction on $k$. Equation \eqref{tower} and Proposition \ref{1_lag} give us that
\begin{align*}
E(X_t|X_{t-2})&=E[E(X_t|X_{t-1})|X_{t-2}]=E\left[\alpha(1-{\alpha_*}^{X_{t-1}})+\mu_{\epsilon}|X_{t-2}\right]\\
&=\alpha[1-E(\alpha_*^{\alpha\minop X_{t-2}+\epsilon_{t-1}}|X_{t-2})]+\mu_{\epsilon}
=\alpha[1-\Psi_{\epsilon_{1}}(\alpha_*)E({\alpha_*}^{\alpha\minop X_{t-2}}|X_{t-2})]+\mu_{\epsilon},
\end{align*} 
with $\alpha_*=\frac{\alpha}{1+\alpha}$.
Using \eqref{gen_1}, we obtain that
\begin{align*}
E(X_t|X_{t-2}=\ell)&=\alpha[1-\Psi_{\epsilon_{1}}(\alpha_*)E({\alpha_*}^{\alpha\minop X_{t-2}}|X_{t-2}=\ell)]+\mu_{\epsilon}\\
&=\alpha\left[1-\Psi_{\epsilon_{1}}(\alpha_*)\frac{1-\alpha(\alpha_*-1)[\alpha_*\alpha/(1+\alpha)]^{\ell}}{1-\alpha(\alpha_*-1)}\right]+\mu_{\epsilon}\\
&=\alpha\left[1-\Psi_{\epsilon_{1}}(\alpha_*)\left(\frac{1+\alpha}{(1+\alpha)^2-\alpha^2}+\frac{\alpha(1+\alpha)-\alpha^2}{(1+\alpha)^2-\alpha^2}{\alpha_*}^{2\ell}\right)\right]+\mu_{\epsilon}=\alpha\left(1-H_2\left({\alpha_*}^{2\ell}\right)\right)+\mu_{\epsilon}.
\end{align*} 

Assume that \eqref{conc_p9} is true for $k=n-1$. Using \eqref{tower}, we have
$$E(X_t|X_{t-n})=E(X_t|\mathcal{F}_{t-n})=E[E(X_t|\mathcal{F}_{t-(n-1)})|\mathcal{F}_{t-n}]=\alpha\left(1-E\left(H_{n-1}\left({\alpha_*}^{(n-1)X_{t-(n-1)}}\right)\middle|X_{t-n}\right)\right)+ \mu_{\epsilon}.$$ 

From the definition of $H_n$, we obtain
$$E\left(H_{n-1}\left({\alpha_*}^{(n-1)X_{t-(n-1)}}\right)\middle|X_{t-n}=\ell\right)=f_2\left(\dots \left(f_{n-2}\left(E\left(f_{n-1}\left({\alpha_*}^{(n-1)X_{t-(n-1)}}\right)\middle|X_{t-n}=\ell\right)\right)\right)\right).$$

Note that 
\begin{align*}E\left(f_{n-1}\left({\alpha_*}^{(n-1)X_{t-(n-1)}}\right)|X_{t-n}=\ell\right)&=h_{n-1}+g_{n-1}E\left[{\alpha_*}^{(n-1)X_{t-(n-1)}}|X_{t-n}=\ell\right]\\
&=h_{n-1}+g_{n-1}E\left[{\alpha_*}^{(n-1)(\alpha\minop X_{t-n})}|X_{t-n}=\ell\right]\\
&=h_{n-1}+g_{n-1}\left[\Psi_{\epsilon_{1}}({\alpha_*}^{n-1})\frac{1-\alpha({\alpha_*}^{n-1}-1)[{\alpha_*}^{n-1}\alpha/(1+\alpha)]^{\ell}}{1-\alpha({\alpha_*}^{n-1}-1)}\right]\\
&=h_{n-1}+g_{n-1}\left[\Psi_{\epsilon_{1}}({\alpha_*}^{n-1})\left(\frac{(1+\alpha)^{n-1}}{(1+\alpha)^n-\alpha^n}+\frac{\alpha(1+\alpha)^{n-1}-\alpha^n}{(1+\alpha)^n-\alpha^n}{\alpha_*}^{n\ell}\right)\right]\\
&=f_{n-1}(f_n({\alpha_*}^{n\ell})).
\end{align*}

Therefore,  
$E\left(H_{n-1}\left({\alpha_*}^{(n-1)X_{t-(n-1)}}\right)\middle|X_{t-n}=\ell\right)=f_2(\dots (f_{n-1}(f_n({\alpha_*}^{nl}))))=H_n({\alpha_*}^{n\ell})$, and hence we get the desired expression
$E(X_t|X_{t-n}=\ell)=\alpha\left(1-H_n\left({\alpha_*}^{n\ell}\right)\right)+ \mu_{\epsilon}$, which completes the proof.
\end{proof}

\begin{prop} Let $h_j$, $g_j$ be as in Proposition \ref{cov_lag_k} and write $\tilde{h}_j=\mu h_j$, for  $j\in\mathbb N$. It holds that
$$\gamma(k):=Cov(X_t,X_{t-k})=\alpha\mu\left[1-H_k(G(\alpha,\mu,k))\right]+\mu\left(\mu_{\epsilon}-\mu\right),$$
where $G(\alpha,\mu,k)=\frac{\alpha_*}{k(1+\mu(1-{\alpha_*}^k))^2}$, and $H_k(\cdot)$ as defined in Proposition \ref{coup}, for $k\in\mathbb N$.
\end{prop}

\begin{proof}Note that
\begin{equation}\label{cov_1}
\gamma(k)=E(E(X_t X_{t-k}|X_{t-k}))-\mu^2=E(X_{t-k}E(X_t| X_{t-k}))-\mu^2=\alpha\left(\mu-E\left(X_{t-k}H_{k}\left({\alpha_*}^{kX_{t-k}}\right)\right)\right)+\mu(\mu_{\epsilon}-\mu),
\end{equation}
where the third equality follows by \eqref{conc_p9}.  A thorough inspection of the definition of $H_k$ gives
\begin{equation}\label{cov_2}
E\left(X_{t-k}H_k\left({\alpha_*}^{kX_{t-k}}\right)\right)=\tilde{f}_2\left(\dots\left(\tilde{f}_k\left(E\left(X_{t-k}{\alpha_*}^{kX_{t-k}}\right)\right)\right)\right),
\end{equation}
where we have defined $\tilde{f}_j(x)=\Psi_{\epsilon_1}({\alpha_*}^{j-1})\left(\tilde{h}_j+g_jx\right)$, for $j\in\mathbb N$.

Note that the argument of the function in \eqref{cov_2} is just a constant times the derivative of $\Psi_{X_1}(s)$ with respect to $s$ and evaluated at ${\alpha_*}$. More specifically, 
\begin{equation}\label{cov_3}
E\left(X_{t-k}{\alpha_*}^{kX_{t-k}}\right)=\frac{\alpha_*}{k}\Psi'_{X_1}({\alpha_*}^{k})=\mu\frac{\alpha_*}{k(1+\mu(1-{\alpha_*}^k))^2}=\mu G(\alpha,\mu,k).
\end{equation}

The second equality follows from \eqref{X_pgf}. Plugging \eqref{cov_3} in \eqref{cov_2}, we obtain
\begin{equation}\label{cov_4}
E\left(X_{t-k}H_k\left({\alpha_*}^{kX_{t-k}}\right)\right)=\tilde{f}_2(\dots(\tilde{f}_k (\mu G(\alpha,\mu,k))))=\mu f_2\left(\dots\left({f}_k \left(G\left(\alpha,\mu,k\right)\right)\right)\right)=\mu H_k(G(\alpha,\mu,k)).
\end{equation}

The result follows by plugging \eqref{cov_4} in \eqref{cov_1}. 
\end{proof}

\section{Parameter estimation}\label{sec:inference}

In this section, we discuss estimation procedures for the geometric MGWI process through conditional least squares (CLS) and maximum likelihood methods. We assume that $X_1,\ldots,X_n$ is a trajectory from the Geo-MGWI model with observed values $x_1,\ldots,x_n$, where $n$ stands for the sample size. We denote the parameter vector by $\boldsymbol\theta \equiv (\mu, \alpha)^\top$.

For the CLS method, we define the function $Q_n(\boldsymbol\theta)$ as
\begin{equation}\label{cls}
Q_n(\boldsymbol\theta) \equiv \sum_{t=2}^n\left\lbrace x_t-\E(X_t\,|\,X_{t-1}=x_{t-1})\right\rbrace^2 =  \sum_{t=2}^n\left\lbrace x_t-\alpha\left[1-\left(\frac{\alpha}{1+\alpha}\right)^{x_{t-1}}\right] - \frac{\mu(1+\mu)}{1+\mu+\alpha}\right\rbrace^2.
\end{equation}

The CLS estimators are obtained as the argument that minimizes $Q_n(\boldsymbol\theta)$, i.e.
\begin{equation}\label{E:thetacls}
\hat{\boldsymbol{\theta}}_{cls} = \argmin_{\boldsymbol\theta} Q_n(\boldsymbol\theta).
\end{equation}
Since we do not have an explicit expression for $\hat{\boldsymbol{\theta}}_{cls}$, numerical optimization methods are required to solve~\eqref{E:thetacls}. This can be done through optimizer packages implemented in softwares such as \texttt{R}~\citep{r2021} and \texttt{MATLAB}. The gradient function associated with $Q_n(\cdot)$ can be provided for these numerical optimizers and is given by
\begin{align*}
\frac{\partial Q_n(\boldsymbol\theta)}{\partial \mu} &= -2\left[1 - \frac{\alpha\left(1+\alpha\right)}{(1+\mu+\alpha)^2}\right] \sum_{t=2}^n \left[x_t - \alpha\left(1-\left(\frac{\alpha}{1+\alpha}\right)^{x_{t-1}}\right) - \frac{\mu(1+\mu)}{1+\mu+\alpha}\right] \quad \text{and}\\
 \frac{\partial Q_n(\boldsymbol\theta)}{\partial \alpha} &= -2 \sum_{t=2}^n \left[x_t - \alpha\left(1-\left(\frac{\alpha}{1+\alpha}\right)^{x_{t-1}}\right) - \frac{\mu(1+\mu)}{1+\mu+\alpha}\right]\left[1 - \left(\frac{\alpha}{1+\alpha}\right)^{x_{t-1}}\left(1+\frac{x_{t-1}}{1+\alpha}\right)  - \frac{\mu(1+\mu)}{(1+\mu+\alpha)^2}\right].
\end{align*}

A strategy to get the standard errors of the CLS estimates based on bootstrap is proposed and illustrated in our empirical illustrations; please see Section \ref{sec:application}. 
 
We now discuss the maximum likelihood estimation (MLE) method. Note that our proposed Geo-MGWI process is a Markov chain (by definition) and therefore the likelihood function can be expressed in terms of the 1-step transition probabilities derived in Proposition~\ref{P:trans_prob}. The MLE estimators are obtained as the argument that maximizes the log-likelihood function, that is, $\hat{\boldsymbol{\theta}}_{mle} = \argmax_{\boldsymbol\theta} \ell_n(\boldsymbol\theta)$, with 
\begin{align}\label{E:loglik}
\ell_n(\boldsymbol\theta) &= \sum_{t=2}^n \log \p(X_t = x_t \,|\,X_{t-1} = x_{t-1})+\log\p(X_1=x_1), 
\end{align}
where the conditional probabilities in $\eqref{E:loglik}$ are given by~\eqref{E:trans_prob} and $\p(X_1=x_1)$ is the probability function of a geometric distribution with mean $\mu$. There is no closed-form expression available for $\hat{\boldsymbol{\theta}}_{mle}$. The maximization of~\eqref{E:loglik} can be accomplished through numerical methods such as the Broyden-Fletcher-Goldfarb-Shanno (BFGS) algorithm implemented in the \texttt{R} package \texttt{optim}. The standard errors of the maximum likelihood estimates can be obtained by using the Hessian matrix associated with (\ref{E:loglik}), which can be evaluated numerically.

In the remaining of this section, we examine and compare the finite-sample behavior of the CLS and MLE methods via Monte Carlo (MC) simulation with 1000 replications per set of parameter configurations, with the parameter estimates computed under both approaches. All the numerical experiments presented in this paper were carried out using the \texttt{R} programming language.

\begin{figure}[!h]
	\centering		
	\includegraphics[scale=0.6]{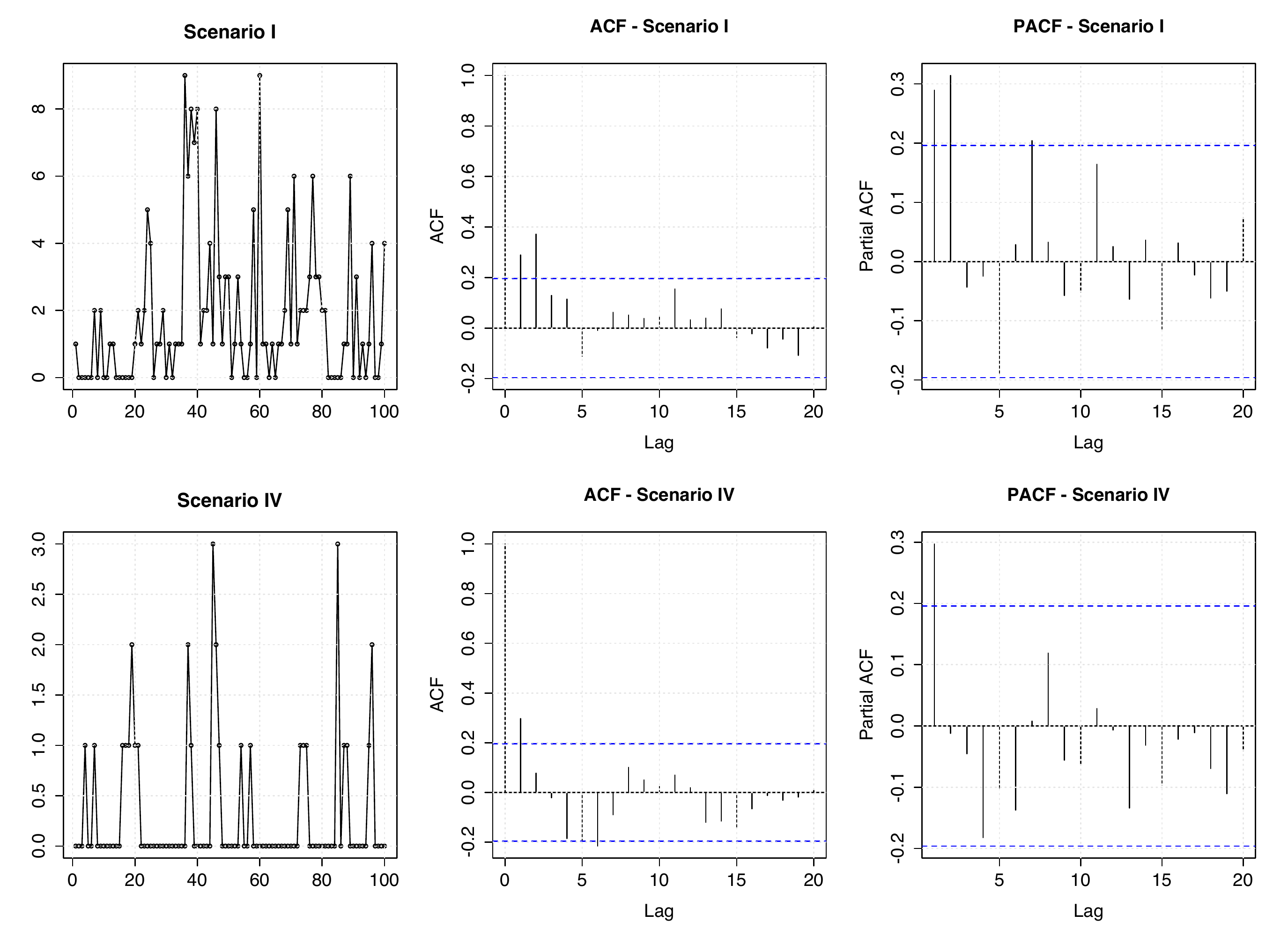}
	\caption{Sample paths for the Geo-MGWI process and their respective ACF and PACF under Scenarios I (top row) and IV (bottom row) with $n=100$.} \label{F:paths_acf_pacf}
\end{figure}

\begin{table}[!h]
	\centering
	\small
	\caption{Empirical mean and RMSE (within parentheses) of the parameter estimates based on the MLE and CLS methods for the Geo-MGWI process under the Scenarios I, II, III, and IV, and for sample sizes $n=100,200,500,1000$.}\label{T:simresults1}\vskip1mm
	\renewcommand{\arraystretch}{1.1}
	\resizebox{\columnwidth}{!}{%
		\tabcolsep=27pt
		\small
		\begin{tabular}{lcccc}
			\toprule	
			\multicolumn{5}{l}{\textbf{Scenario I}: $\mu = 2.0, \,\alpha = 1.0$}\\
			\midrule 													
			$n$ & $\hat\mu_{mle}$ & $\hat\alpha_{mle}$ & $\hat\mu_{cls}$ & $\hat\alpha_{cls}$ \\
			\midrule	
			100 & 1.996 (0.281) & 0.957 (0.425) & 1.999 (0.282) & 0.962 (0.698) \\
			200 & 1.995 (0.197) & 1.021 (0.290) & 1.998 (0.197) & 1.059 (0.536) \\
			500 & 2.013 (0.124) & 0.987 (0.177) & 2.014 (0.125) & 0.991 (0.339) \\
			1000 & 1.998 (0.088) & 0.998 (0.128) & 1.998 (0.088) & 0.988 (0.238) \\
			\midrule
			\multicolumn{5}{l}{\textbf{Scenario II}: $\mu = 1.2, \,\alpha = 0.5$}\\
			\midrule 													
			$n$ & $\hat\mu_{mle}$ & $\hat\alpha_{mle}$ & $\hat\mu_{cls}$ & $\hat\alpha_{cls}$ \\
			\midrule
			100 & 1.206 (0.181) & 0.486 (0.289) & 1.208 (0.182) & 0.556 (0.482) \\
			200 & 1.197 (0.128) & 0.491 (0.205) & 1.198 (0.129) & 0.495 (0.327) \\
			500 & 1.196 (0.082) & 0.498 (0.119) & 1.196 (0.082) & 0.490 (0.197) \\
			1000 & 1.200 (0.058) & 0.506 (0.090) & 1.200 (0.058) & 0.494 (0.143) \\
			\midrule
			\multicolumn{5}{l}{\textbf{Scenario III}: $\mu = 0.5, \,\alpha = 1.5$}\\
			\midrule 													
			$n$ & $\hat\mu_{mle}$ & $\hat\alpha_{mle}$ & $\hat\mu_{cls}$ & $\hat\alpha_{cls}$ \\
			\midrule
			100 & 0.499 (0.130) & 1.515 (0.523) & 0.498 (0.132) & 1.487 (0.831) \\
			200 & 0.499 (0.091) & 1.514 (0.387) & 0.498 (0.093) & 1.495 (0.595) \\
			500 & 0.496 (0.058) & 1.490 (0.236) & 0.496 (0.059) & 1.472 (0.356) \\
			1000 & 0.500 (0.042) & 1.502 (0.174) & 0.500 (0.044) & 1.524 (0.299) \\
			\midrule
			\multicolumn{5}{l}{\textbf{Scenario IV}: $\mu = 0.3, \,\alpha = 0.5$}\\
			\midrule 													
			$n$ & $\hat\mu_{mle}$ & $\hat\alpha_{mle}$ & $\hat\mu_{cls}$ & $\hat\alpha_{cls}$ \\
			\midrule
			100 & 0.298 (0.078) & 0.506 (0.271) & 0.298 (0.078) & 0.504 (0.340) \\
			200 & 0.296 (0.057) & 0.491 (0.186) & 0.297 (0.057) & 0.492 (0.244) \\
			500 & 0.299 (0.037) & 0.496 (0.120) & 0.300 (0.037) & 0.504 (0.157) \\
			1000 & 0.299 (0.026) & 0.499 (0.087) & 0.299 (0.026) & 0.500 (0.110) \\
			\bottomrule
	\end{tabular}}
\end{table}

We consider four simulation scenarios with different values for $\boldsymbol\theta = (\mu, \alpha)^\top$, namely: (I) $\boldsymbol\theta = (2.0, 1.0)^\top$, (II) $\boldsymbol\theta = (1.2, 0.5)^\top$, (III) $\boldsymbol\theta = (0.5, 1.5)^\top$, and (IV) $\boldsymbol\theta = (0.3,0.5)^\top$. To illustrate these configurations, we display in Figure~\ref{F:paths_acf_pacf} simulated trajectories from the Geo-MGWI process and their associated ACF and PACF under Scenarios I and IV. In Table~\ref{T:simresults1}, we report the empirical mean and root mean squared error (RMSE) of the parameter estimates obtained from the MC simulation based on the MLE and CLS methods. We can observe that both approaches produce satisfactory results and also a slight advantage of the MLE estimators over the CLS for estimating $\alpha$, mainly in terms of RMSE, which is already expected. This advantage can also be seen from Figure~\ref{F:boxplots_s1_s4}, which presents boxplots of the parameter estimates for $\mu$ and $\alpha$ under the Scenarios I and IV with sample sizes $n = 100$ and $n = 500$. In general, the estimation procedures considered here produce estimates with bias and RMSE decreasing towards zero as the sample size increases, therefore giving evidence of consistency.

\begin{figure}[!h]
	\centering		
	\includegraphics[scale=0.6]{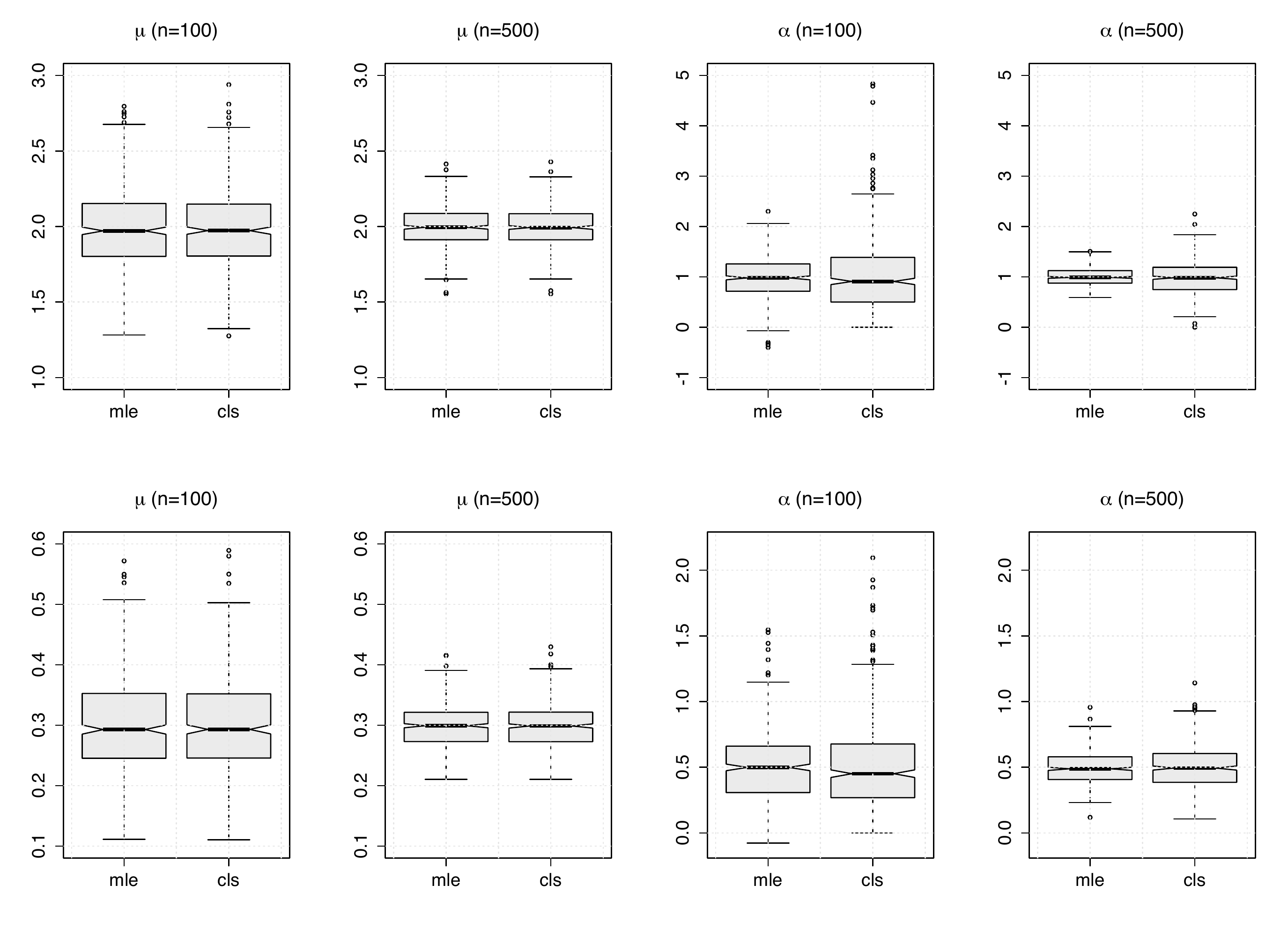}
	\caption{Boxplots of the MLE and CLS estimates for the Geo-MGWI process under the Scenarios I (top row) and IV (bottom row), and for sample sizes $n=100,500$.} \label{F:boxplots_s1_s4}
\end{figure}

\section{Dealing with non-stationarity}\label{sec:nonstat}

In many practical situations, stationarity can be a non-realistic assumption; for instance, see \cite{bra1995}, \cite{encetal2009}, and \cite{wan2020} for works that investigate non-stationary Poisson INAR process. Motivated by that, in this section, we propose a non-stationary version of the Geo-MGWI process allowing for time-varying parameters. Consider 
\begin{eqnarray*}
\mu_t=\exp({\bf w}_t^\top\boldsymbol\beta)\quad\mbox{and}\quad \alpha_t=\exp({\bf v}_t^\top\boldsymbol\gamma),
\end{eqnarray*}
where ${\bf w}_t$ and ${\bf v}_t$ are $p\times 1$ and $q\times 1$ covariate vectors for $t\geq 1$, and $\boldsymbol\beta$ and $\boldsymbol\gamma$ are $p\times 1$ and $q\times 1$ vectors of associated regression coefficients.

We define a time-varying or non-stationary Geo-MGWI process by 
\begin{eqnarray}\label{nonstatgeo}
X_t=\alpha_t\minop X_{t-1}+\epsilon_t, \,\,\,t=2,3,\ldots,
\end{eqnarray}
and $X_1\sim\mbox{Geo}(\mu_1)$, where $\alpha_t\minop X_{t-1}=\min(X_{t-1},Z_t)$, $\{Z_t\}_{t\in\mathbb{N}}$ is an independent sequence with $Z_t\sim\mbox{Geo}(\alpha_t)$, $\{\epsilon_t\}_{t\geq 1}$ are independent random variables with $\epsilon_t\sim\mbox{ZMG}\left(\dfrac{\alpha_t}{1+\alpha_t+\mu_t},\mu_t\right)$, for $t\geq2$. It is also assumed that $\epsilon_t$ is independent of $X_{t-l}$ and $Z_{t-l+1}$, for all $l\geq 1$. Under these assumptions, the marginals of the process (\ref{nonstatgeo}) are $\mbox{Geo}(\mu_t)$ distributed, for $t\in\mathbb N$.

We consider two estimation methods for the parameter vector $\boldsymbol\theta=(\boldsymbol\beta,\boldsymbol\gamma)^\top$. The first one is based on the conditional least squares. The CLS estimator of $\boldsymbol\theta$ is obtained by minimizing (\ref{cls}) with $\mu_t$ and $\alpha_t$ instead of $\mu$ and $\alpha$, respectively. According to \cite{wan2020}, this procedure might not be accurate in the sense that non-significant covariates can be included in the model. In that paper, a penalized CLS (PCLS) method is considered. Hence, a more accurate estimator is obtained by minimizing $\widetilde Q_n(\boldsymbol\theta)=Q_n(\boldsymbol\theta)+n\sum_{j=1}^{p+q} P_\delta(|\theta_i|)$, where $P_\delta(\cdot)$ is a penalty function and $\delta$ is a tuning parameter. See \cite{wan2020} for possible choices of penalty function. This can be used as a selection criterion and we hope to explore it in a future paper. A second method for estimating the parameters is the maximum likelihood method. The log-likelihood function assumes the form (\ref{E:loglik}) with $\mu$ and $\alpha$ replaced by $\mu_t$ and $\alpha_t$, respectively.

For the non-stationary case, we carry out a second set of Monte Carlo simulations by considering trend and seasonal covariates in the model as follows: $$\mu_t = \exp(\beta_0 + \beta_1 t/n + \beta_2 \cos (2\pi t/12)) \quad \mbox{and} \quad \alpha_t = \exp(\gamma_0 + \gamma_1 t/n),$$ for $t = 1, \ldots, n$. The above structure aims to mimic realistic situations when dealing with epidemic diseases. We here set the following scenarios: (V) $(\beta_0, \beta_1, \beta_2, \gamma_0, \gamma_1) = (2.0, 1.0, 0.7, 2.0, 1.0)$ and (VI) $(\beta_0, \beta_1, \beta_2, \gamma_0, \gamma_1) = (3.0, 1.0, 0.5, 3.0, 2.0)$.  We consider 500 Monte Carlo replications and the sample sizes $n = 100, 200, 500, 1000$. Table~\ref{T:simresults2} reports the empirical mean and the RMSE (within parentheses) of the parameter estimates based on the MLE and CLS methods. We can observed that the MLE method outperforms the CLS method for all configurations considered, as expected since we are generating time series data from the ``true" model. This can be also seen from Figure \ref{F:boxplots_s5}, which presents the boxplots of MLE and CLS estimates under the Scenarios V with sample sizes $n=200,500$. Regardless, note that the bias and RMSE of the CLS estimates decrease as the sample size increases.

\begin{table}[!h]
  \centering
  \small
  \caption{Empirical mean and RMSE (within parentheses) of the parameter estimates based on the MLE and CLS methods for the non-stationary Geo-MGWI process under the Scenarios V and VI, and for sample sizes $n=100,200,500,1000$.}\label{T:simresults2}\vskip1mm
  \renewcommand{\arraystretch}{1.2}
  \begin{tabular}{llllcccccccccc}
  \toprule
  \multicolumn{14}{l}{\textbf{Scenario V}\,\, $(\beta_0, \beta_1, \beta_2, \gamma_0, \gamma_1) = (2.0, 1.0, 0.7, 2.0, 1.0)$} \\
  \midrule
   &&&& $\hat\beta_0$ && $\hat\beta_1$ && $\hat\beta_2$ && $\hat\gamma_0$ && $\hat\gamma_1$ \\
   \cmidrule{5-14}
   $n = 100$ && MLE && 1.969 (0.278) && 1.006 (0.473) && 0.662 (0.209) && 1.939 (0.635) && 1.027 (0.741)\\
   			 && CLS && 1.398 (2.425) && 1.373 (1.997) && 0.914 (1.317) && 1.777 (1.649) && 0.469 (2.531)\\
   \midrule
   $n = 200$ && MLE && 1.984 (0.206) && 0.987 (0.337) && 0.677 (0.138) && 1.986 (0.453) && 0.963 (0.503)\\
			 && CLS && 1.743 (1.212) && 1.117 (1.005) && 0.813 (0.725) && 1.783 (1.295) && 0.732 (1.557)\\
   \midrule
   $n = 500$ && MLE && 1.993 (0.126) && 1.007 (0.209) && 0.673 (0.091) && 2.005 (0.168) && 1.003 (0.272)\\ 
             && CLS && 1.923 (0.317) && 1.059 (0.376) && 0.706 (0.241) && 1.864 (0.622) && 1.184 (1.000)\\
   \midrule
   $n = 1000$ && MLE && 1.996 (0.083) && 1.006 (0.136) && 0.674 (0.064) && 2.012 (0.116) && 0.999 (0.200)\\ 
             && CLS && 1.929 (0.265) && 1.061 (0.284) && 0.696 (0.205) && 1.869 (0.461) && 1.230 (0.665)\\          
  \midrule
  \multicolumn{14}{l}{\textbf{Scenario VI}\,\, $(\beta_0, \beta_1, \beta_2, \gamma_0, \gamma_1) = (3.0, 1.0, 0.5, 3.0, 2.0)$ }\\
  \midrule
   &&&& $\hat\beta_0$ && $\hat\beta_1$ && $\hat\beta_2$ && $\hat\gamma_0$ && $\hat\gamma_1$ \\
   \cmidrule{5-14}
   $n = 100$ && MLE && 2.967 (0.296) && 0.965 (0.554) && 0.506 (0.221) && 2.935 (0.376) && 2.049 (0.629)\\ 
             && CLS && 2.283 (1.732) && 1.264 (1.960) && 0.704 (1.098) && 3.231 (1.290) && 1.291 (2.048)\\   
   \midrule
   $n = 200$ && MLE && 2.981 (0.198) && 0.995 (0.363) && 0.484 (0.152) && 2.996 (0.249) && 1.998 (0.435)\\
             && CLS && 2.357 (1.390) && 1.310 (1.410) && 0.601 (0.884) && 3.561 (1.165) && 1.106 (1.855)\\
   \midrule
   $n = 500$ && MLE && 2.996 (0.121) && 0.988 (0.226) && 0.484 (0.093) && 3.008 (0.147) && 1.980 (0.264)\\
             && CLS && 2.570 (0.950) && 1.279 (0.912) && 0.641 (0.526) && 3.370 (0.886) && 1.481 (1.354)\\
   \midrule          
   $n = 1000$ && MLE && 2.998 (0.083) && 1.004 (0.156) && 0.477 (0.067) && 2.999 (0.099) && 2.009 (0.184)\\
             && CLS && 2.697 (0.623) && 1.192 (0.597) && 0.601 (0.452) && 3.244 (0.714) && 1.737 (1.026) \\          
  \bottomrule 
  \end{tabular}
\end{table}

\begin{figure}[!h]
	\centering		
	\includegraphics[scale=0.55]{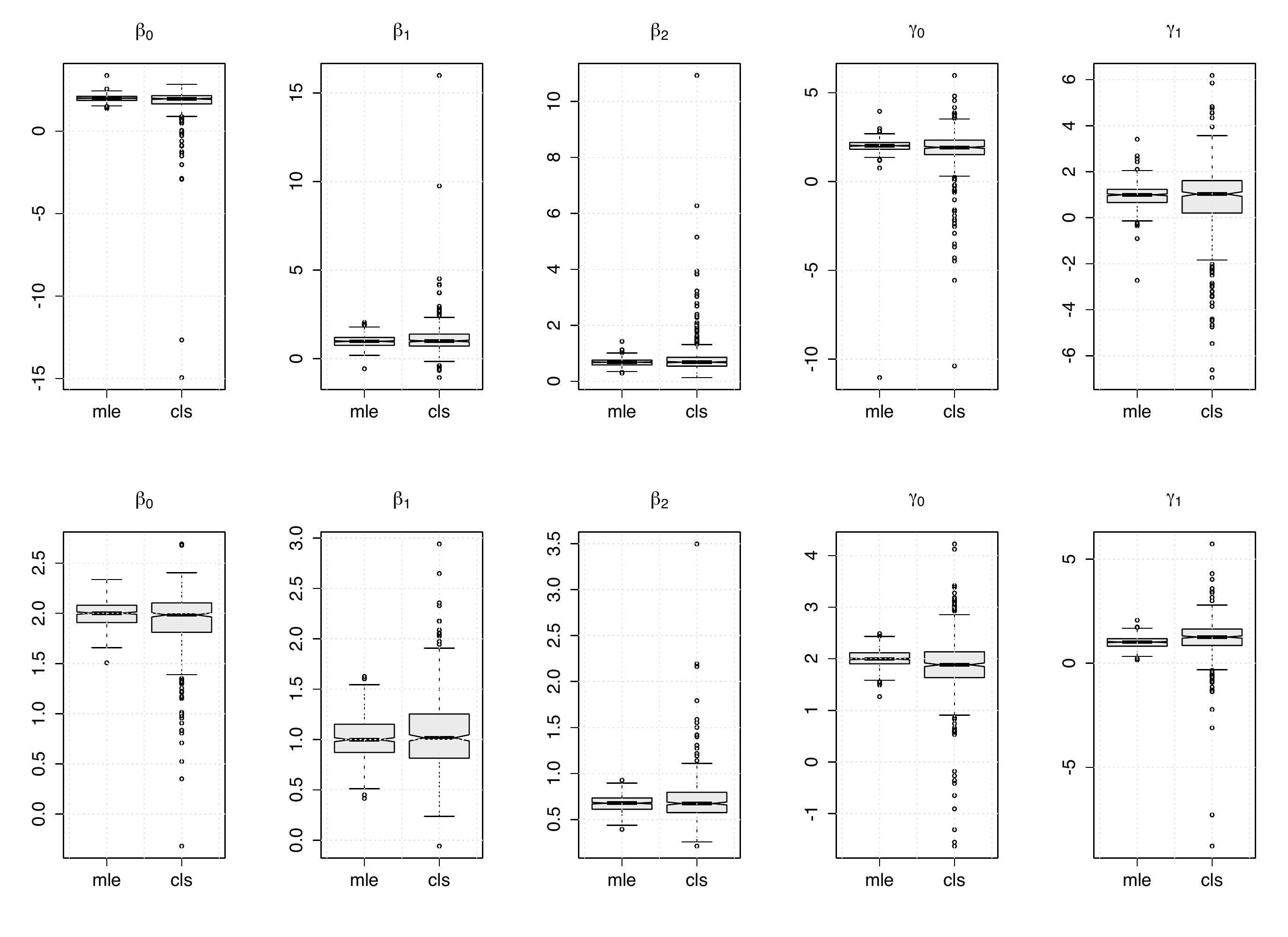}
	\caption{Boxplots of MLE and CLS estimates for the non-stationary Geo-MGWI process under the Scenarios V with sample sizes $n = 200$ (top row) and $n = 500$ (bottom row).} \label{F:boxplots_s5}
\end{figure}

\section{Real data applications}\label{sec:application}

In this section, we discuss the usefulness of our methodology under stationary and non-stationary conditions. In the first empirical example, we consider the monthly number of polio cases reported to the U.S. Centers for Disease Control and Prevention from January 1970 to December 1983, with 168 observations. The data were obtained through the \texttt{gamlss} package in \texttt{R}. Polio (or poliomyelitis) is a disease caused by \textit{poliovirus}. Symptoms associated with polio can vary from mild flu-like symptoms to paralysis and possibly death, mainly affecting children under 5 years of age. The second example concerns the monthly number of Hansen's disease cases in the state of Paraíba, Brazil, reported by DATASUS - Information Technology Department of the Brazilian Public Health Care System (SUS), from January 2001 to December 2020, totalizing 240 observations. Hansen's disease (or leprosy) is a curable infectious disease that is caused by \textit{M. leprae}. It mainly affects the skin, the peripheral nerves mucosa of the upper respiratory tract, and the eyes. According to the World Health Organization, about 208000 people worldwide are infected with Hansen's disease. The data are displayed in Table~\ref{T:hans_data}.

\begin{table}[!h]
	\centering
	\footnotesize
	\caption{The monthly cases of Hansen's disease in the state of Paraíba, Brazil.}\label{T:hans_data}
	\begin{tabular}{rrrrrrrrrrrrr}
		\toprule
		& Jan & Feb & Mar & Apr & May & Jun & Jul & Aug & Sep & Oct & Nov & Dec \\ 
		\midrule
		\texttt{2001} &  60 &  58 &  91 &  72 &  94 &  52 &  54 &  78 &  64 & 111 &  81 &  70 \\ 
  		\texttt{2002} &  55 &  72 &  71 &  70 &  61 &  51 &  80 &  82 &  97 & 107 & 142 &  81 \\ 
  		\texttt{2003} &  92 & 106 & 126 &  78 &  86 &  69 &  91 &  91 &  64 &  83 &  83 &  55 \\ 
  		\texttt{2004} &  67 &  82 & 121 &  84 & 102 &  77 &  83 & 102 &  77 &  59 &  86 &  67 \\ 
  		\texttt{2005} &  59 &  86 &  84 & 102 &  75 &  57 &  82 & 126 & 107 & 123 & 138 &  94 \\ 
  		\texttt{2006} &  88 &  78 & 105 &  91 & 106 &  68 &  85 & 106 &  95 &  80 & 101 &  67 \\ 
  		\texttt{2007} &  78 &  81 &  96 &  68 &  94 &  67 &  66 &  88 &  71 &  84 &  74 &  64 \\ 
  		\texttt{2008} &  79 &  75 &  66 &  81 &  74 &  45 &  82 &  91 &  85 &  74 &  77 &  61 \\ 
  		\texttt{2009} &  53 &  79 & 105 &  81 &  68 &  67 &  64 &  73 &  75 &  76 &  85 &  48 \\ 
  		\texttt{2010} &  51 &  74 &  94 &  64 &  60 &  51 &  54 &  70 &  69 &  68 &  64 &  43 \\ 
  		\texttt{2011} &  66 &  67 &  83 &  77 &  71 &  67 &  58 &  90 &  73 &  59 &  78 &  72 \\ 
  		\texttt{2012} &  71 &  82 &  80 &  64 &  82 &  60 &  83 &  77 &  76 &  60 &  49 &  52 \\ 
  		\texttt{2013} &  54 &  53 &  80 &  83 &  52 &  52 &  79 &  61 &  71 &  61 &  78 &  47 \\ 
  		\texttt{2014} &  61 &  79 &  51 &  63 &  51 &  45 &  61 &  63 &  83 &  63 &  60 &  40 \\ 
  		\texttt{2015} &  64 &  53 &  79 &  43 &  55 &  47 &  48 &  66 &  48 &  48 &  46 &  48 \\ 
  		\texttt{2016} &  39 &  43 &  54 &  34 &  50 &  38 &  38 &  67 &  35 &  44 &  48 &  41 \\ 
  		\texttt{2017} &  40 &  46 &  54 &  43 &  43 &  53 &  45 &  68 &  65 &  44 &  58 &  47 \\ 
  		\texttt{2018} &  64 &  42 &  72 &  62 &  51 &  42 &  43 &  64 &  47 &  48 &  76 &  40 \\ 
  		\texttt{2019} &  63 &  70 &  56 &  54 &  59 &  51 &  60 &  65 &  80 &  85 &  65 &  49 \\ 
  		\texttt{2020} &  57 &  62 &  61 &  16 &  21 &  19 &  35 &  25 &  60 &  63 &  51 &  30 \\
  		\texttt{2021} &  35 &  53 &  56 &  41 &  44 &  41 &  32 &  33 &  17 &   5 &   5 &   5 \\ 
  		\bottomrule
	\end{tabular}
\end{table} 

\subsection{Polio data analysis}

We begin the analysis of the polio data by providing plots of the observed time series and the corresponding sample ACF and PACF plots in Figure~\ref{F:polio_data}. These plots give us evidence that the count time series is stationary. Table~\ref{T:desc_polio} provides a summary of the polio data with  descriptive statistics, including mean, median, variance, skewness, and kurtosis. From the results in Table~\ref{T:desc_polio}, we can observe that counts vary between 0 and 14, with the sample mean and variance equal to 1.333 and 3.505, respectively, which suggests overdispersion of the data. 
 
\begin{figure}[!h]
	\centering		
	\includegraphics[scale=0.5]{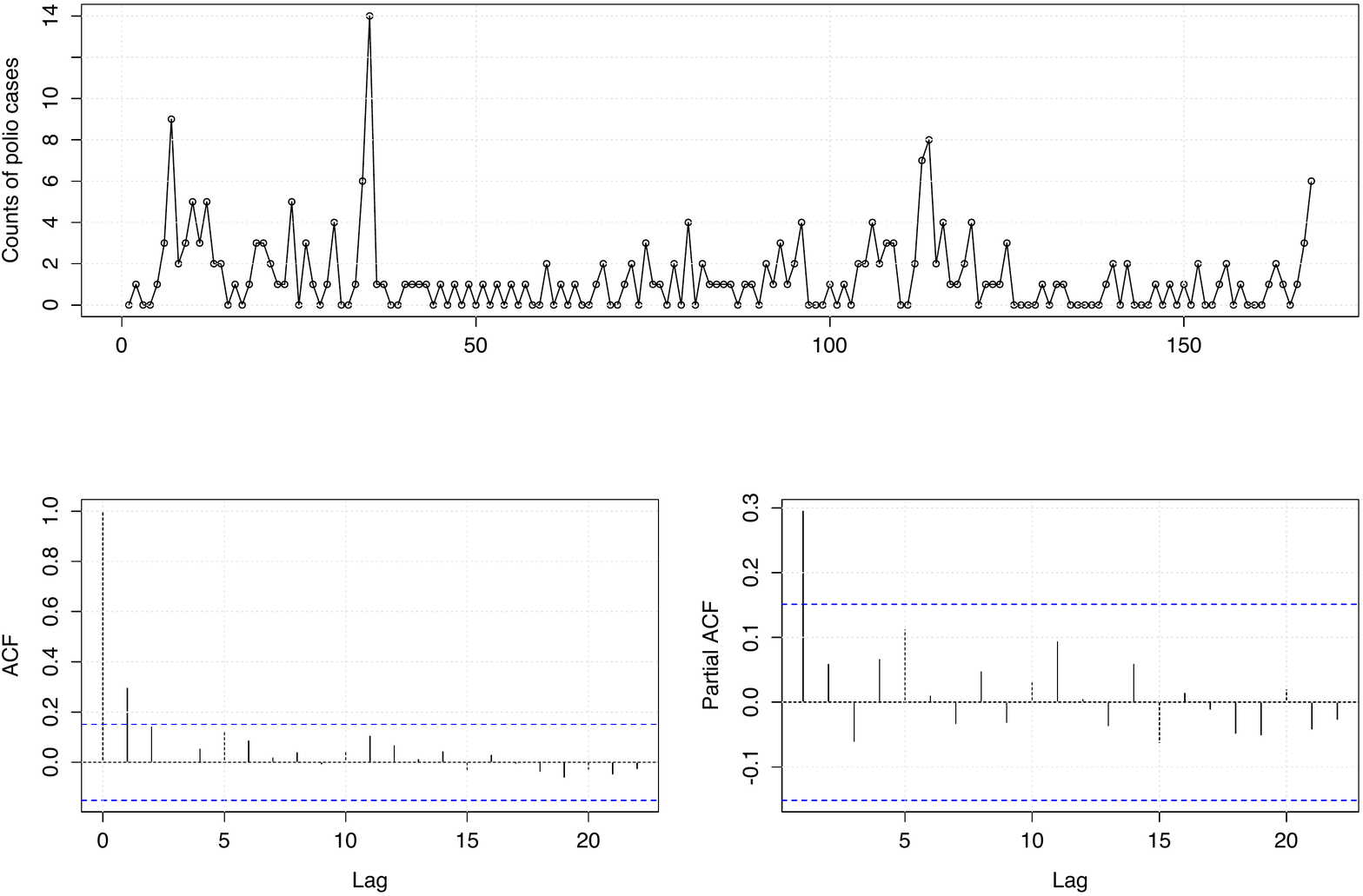}
	\caption{Polio data (top panel) and corresponding autocorrelation function (bottom left panel) and partial autocorrelation function (bottom right panel).} \label{F:polio_data}
\end{figure} 

\begin{table}[!h]
\centering
\small
\caption{Descriptive statistics of the polio data.} \label{T:desc_polio}
 \begin{tabular}{cccccccccccccc}
 \toprule
 Minimum && Maximum && Mean && Median && Variance && Skewness && Kurtosis\\
 \midrule
 0 && 14 && 1.333 && 1.000 && 3.505 && 3.052 && 16.818\\
 \bottomrule
 \end{tabular}
\end{table}

For comparison purposes, we consider the classic first-order GWI/INAR process with $E(X_t|X_{t-1})=\alpha X_{t-1}+\mu(1-\alpha)$. This linear conditional expectation on $X_{t-1}$ holds for the classic stationary INAR processes such as the binomial thinning-based ones, in particular, the Poisson INAR(1) model by \cite{alzalo1987}. The aim is to evaluate the effect of the nonlinearity of our proposed models on the prediction in comparison to the classic GWI/INAR(1) processes.

We consider the CLS estimation procedure, where just the conditional expectation is considered. This allows for a more flexible approach since no further assumptions are required. To obtain the standard errors of the CLS estimates, we consider a parametric bootstrap where some model satisfying the specific form for the conditional expectation holds. In this first application, for our MGWI process, we consider the geometric model derived in Section \ref{sec:gwi_process}. For the classic INAR, the Poisson model by \cite{alzalo1987} is considered in the bootstrap approach. This strategy to get standard errors has been considered, for example, by \cite{maietal2021} for a class of semiparametric time series models driven by a latent factor. In order to compare the predictive performance of the competing models, we compute the sum of squared prediction errors (SSPE) defined by $\text{SSPE} = \sum_{t = 2}^n(x_t - \hat\mu_t)^2$, 
where $\hat\mu_t = \widehat E(X_t|X_{t-1})$ is the predicted mean at time $t$, for $t = 2, \ldots, n$. Table~\ref{T:results_polio} summarizes the fitted models by providing CLS estimates and their respective standard errors, and the SSPE values. The SSPE results in Table~\ref{T:results_polio} show the superior performance of the MGWI process over the classic GWI/INAR process in terms of prediction. This can also be observed from Figure~\ref{F:app_polio}, where the MGWI process shows a better agreement between the observed and predicted values.
\begin{table}[!h]
\centering
\small
\renewcommand{\arraystretch}{1.2}
\caption{\small CLS estimates, standard errors, and SSPE values of the fitted MGWI and classic GWI/INAR model for the monthly cases of polio.}\label{T:results_polio}
  \begin{tabular}{cccccccccc}
  \toprule
  Models && Parameters && Estimates && Stand. Errors && SSPE\\
  \midrule
  \multirow{ 2}{*}{MGWI} && $\mu$ && 1.3585 && 0.2047 && \multirow{ 2}{*}{522.8987} \\
  	   && $\alpha$ && 2.6514 && 1.2230 &&  \\
  \midrule	   
   \multirow{ 2}{*}{GWI/INAR} && $\mu$ && 1.3572 && 0.1627 && \multirow{ 2}{*}{530.6749}\\
        &&  $\alpha$ && 0.3063 && 0.0772 && \\
  \bottomrule      
  \end{tabular}
\end{table}

\begin{figure}[!h]
	\centering		
	\includegraphics[scale=0.55]{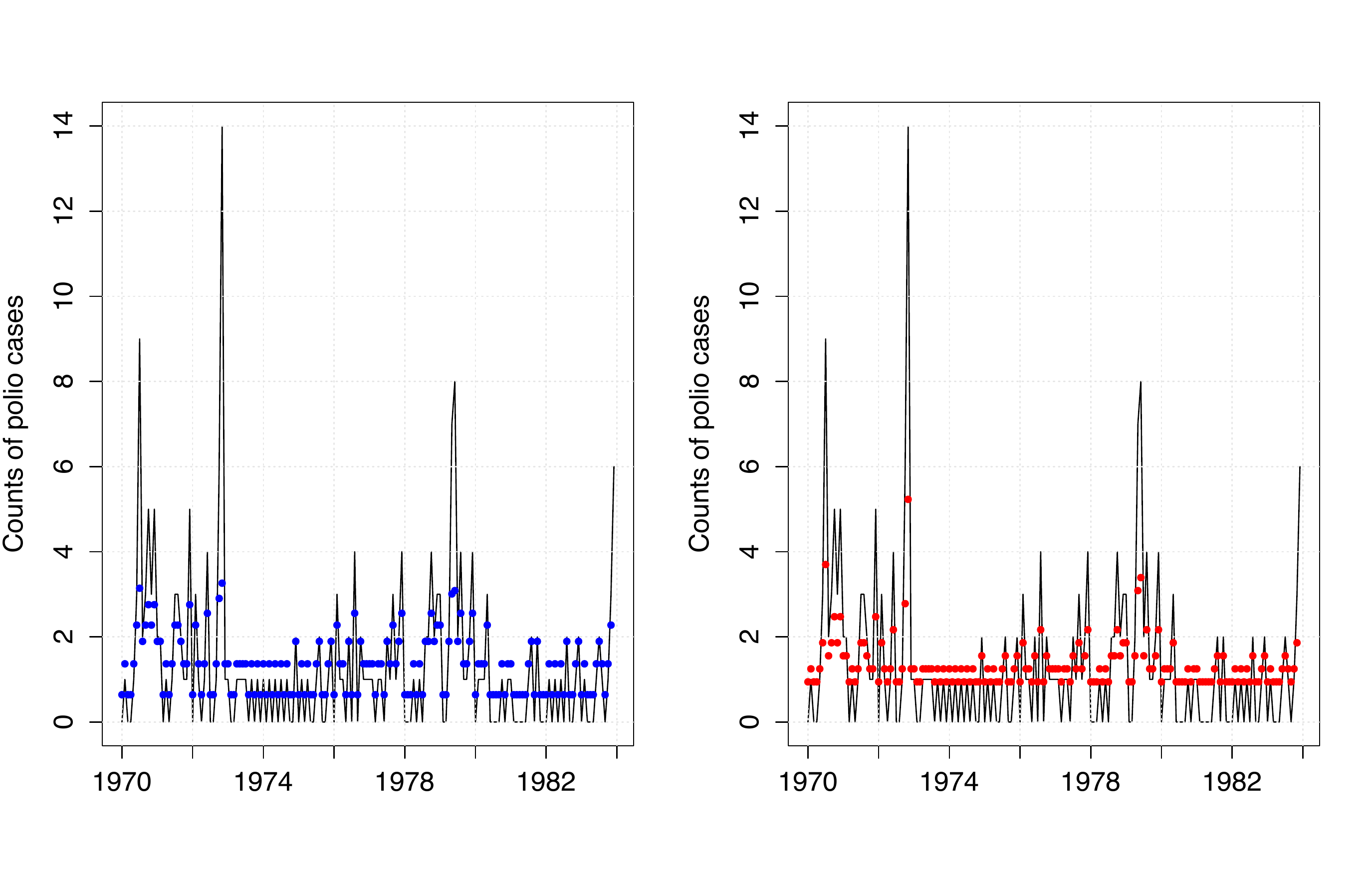}
	\vspace{-5mm}\caption{Plots of polio data (solid lines) and fitted conditional means (dots) based on the MGWI process (to the left) and classic GWI/INAR process (to the right).} \label{F:app_polio}
\end{figure}

To evaluate the adequacy of our proposed MGWI process, we consider the Pearson residuals defined by $R_t \equiv (X_t - \hat\mu_t) / \hat\sigma_t$, where $\hat\sigma_t = \sqrt{\widehat{\mbox{Var}}(X_t|X_{t-1})}$, for $t = 2, \ldots, n$, where we assume that the conditional variance takes the form given in Proposition \ref{1_lag}. Figure \ref{F:residuals_polio} presents the Pearson residuals against the time, its ACF, and the qq-plot against the normal quantiles. These plots show that the data correlation was well-captured. On the other hand, the qq-plot suggests that the Pearson residuals are not normally distributed. Actually, this discrepancy is not unusual especially when dealing with low counts; for instance, see \cite{zhu2011} and \cite{silbar2019}. As an alternative way to check the adequacy, we use the normal pseudo-residuals introduced by \cite{dunsmy1996}, which is defined by $R^\ast_t = \Phi^{-1}(U_t)$, where $\Phi(\cdot)$ is the standard normal distribution function and $U_t$ is uniformly distributed on the interval $(F_{\boldsymbol{\hat\theta}}(x_t-1), F_{\boldsymbol{\hat\theta}}(x_t))$, where $F_{\boldsymbol{\hat\theta}}(\cdot)$ is the fitted predictive cumulative distribution function of the MGWI process. Figure~\ref{F:residuals_polio2} shows the pseudo residuals against the time, its ACF, and qq-plot. We can observe that the pseudo-residuals are not correlated and are approximately normally distributed. Therefore, we conclude that the MGWI process provides an adequate fit to the polio count time series data.

\begin{figure}[!h]
	\centering		
	\includegraphics[scale=0.5]{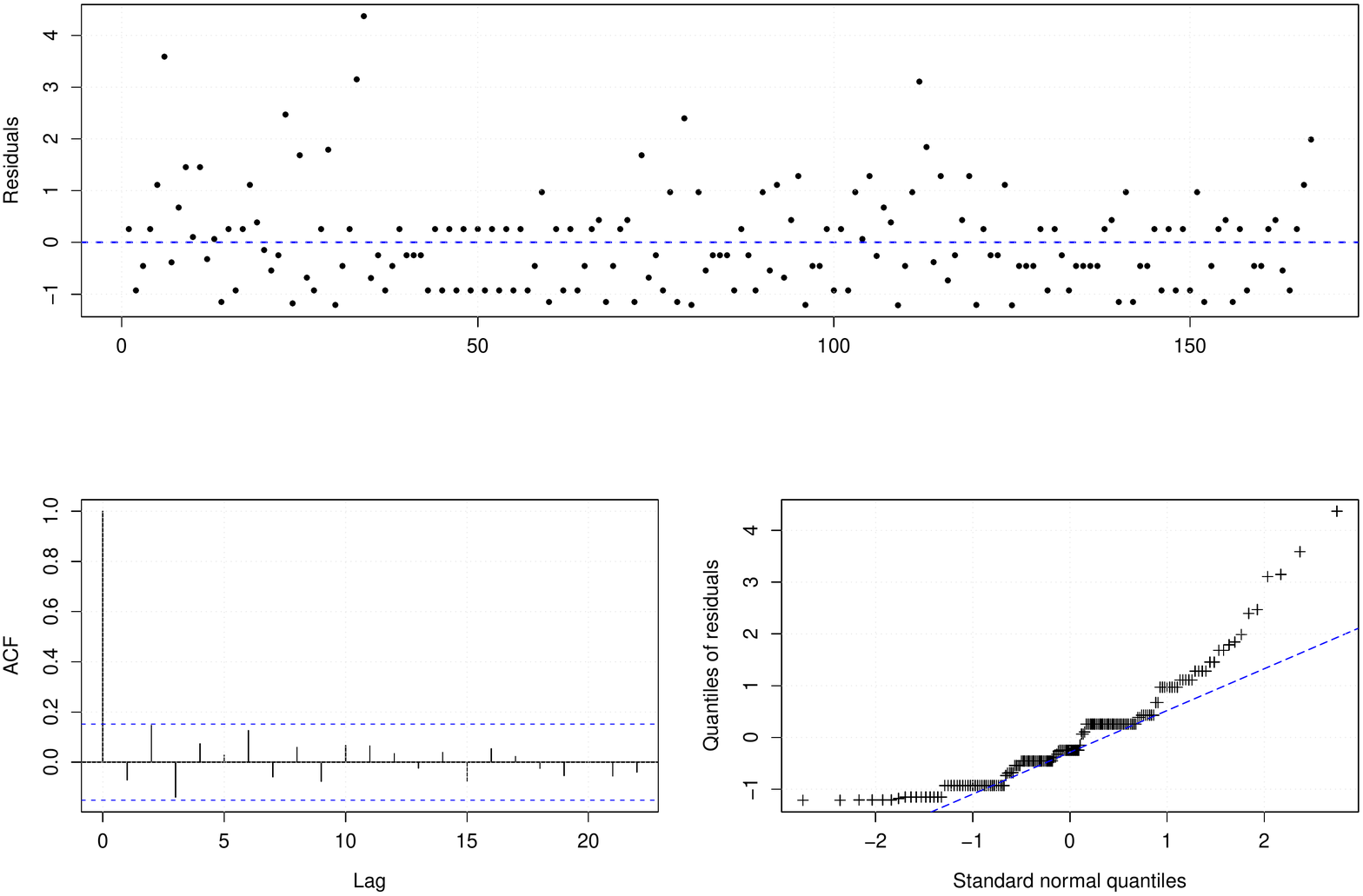}
	\caption{Pearson residuals for the MGWI process fitted to the polio data: residuals against time (top panel), ACF (bottom left panel) and qq-plot (bottom right panel).
} \label{F:residuals_polio}
\end{figure}

\begin{figure}[!h]
	\centering		
	\includegraphics[scale=0.5]{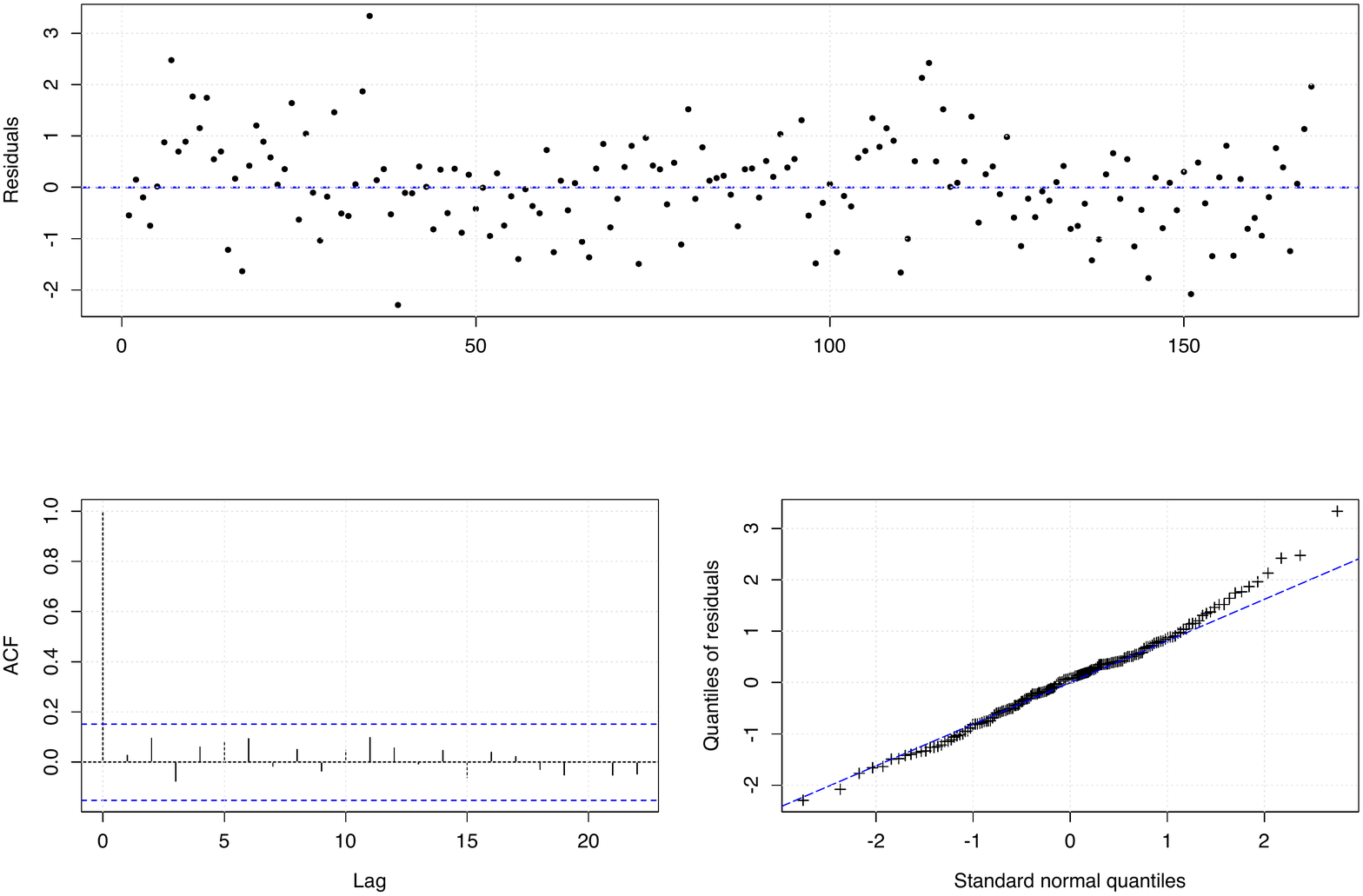}
	\caption{Pseudo-residuals for the MGWI process fitted to the polio data: residuals against time (top panel), ACF (bottom left panel) and qq-plot (bottom right panel).
} \label{F:residuals_polio2}
\end{figure}

\subsection{Hansen's disease data analysis}

We now analyze Hansen's disease data. A descriptive data analysis is provided in Table~\ref{T:desc_hans}. Figure~\ref{F:hans_data} presents the Hansen's count data and its corresponding sample ACF and PACF plots. This figure provides evidence that the count time series is non-stationary. In particular, we can observe a negative trend. This motivates us to use non-stationarity approaches to handle this data. We consider our non-stationary MGWI process with conditional mean 
\begin{eqnarray}\label{condmean_nonstat}
E(X_t|X_{t-1})=\alpha_t\left[1-\left(\dfrac{\alpha_t}{1+\alpha_t}\right)^{X_{t-1}}\right]+\dfrac{\mu_t(1+\mu_t)}{1+\mu_t+\alpha_t},
\end{eqnarray}
where the following regression structure is assumed:
\begin{align*}
\mu_t = \exp \left(\beta_0 + \beta_1t/252 \right) \quad \text{and} \quad \alpha_t = \exp \left(\gamma_0 + \gamma_1t/252 \right), \,\,\, \text{for}\,\,\, t = 1, \ldots, 252,
\end{align*}
with the term $t/252$ being a linear trend. For comparison purposes, we also consider the Poisson INAR(1) process allowing for covariates \citep{bra1995} with conditional expectation $E(X_t|X_{t-1})=\alpha_tX_{t-1}+\mu_t(1-\alpha_t)$, where 
\begin{align*}
\mu_t = \exp \left(\beta_0 + \beta_1t/252 \right) \quad \text{and} \quad \alpha_t = \dfrac{\exp \left(\gamma_0 + \gamma_1t/252 \right)}{1+\exp \left(\gamma_0 + \gamma_1t/252 \right)}, \,\,\, \text{for}\,\,\, t = 1, \ldots, 252.
\end{align*}

\begin{table}[!h]
\centering
\small
\caption{Descriptive statistics of the Hansen's disease data.} \label{T:desc_hans}
 \begin{tabular}{cccccccccccccc}
 \toprule
 Minimum && Maximum && Mean && Median && Variance && Skewness && Kurtosis\\
 \midrule
 5 && 142 && 66.63 && 66 && 481.103 && 0.250 && 3.937\\
 \bottomrule
 \end{tabular}
\end{table}

\begin{figure}[!h]
	\centering		
	\includegraphics[scale=0.5]{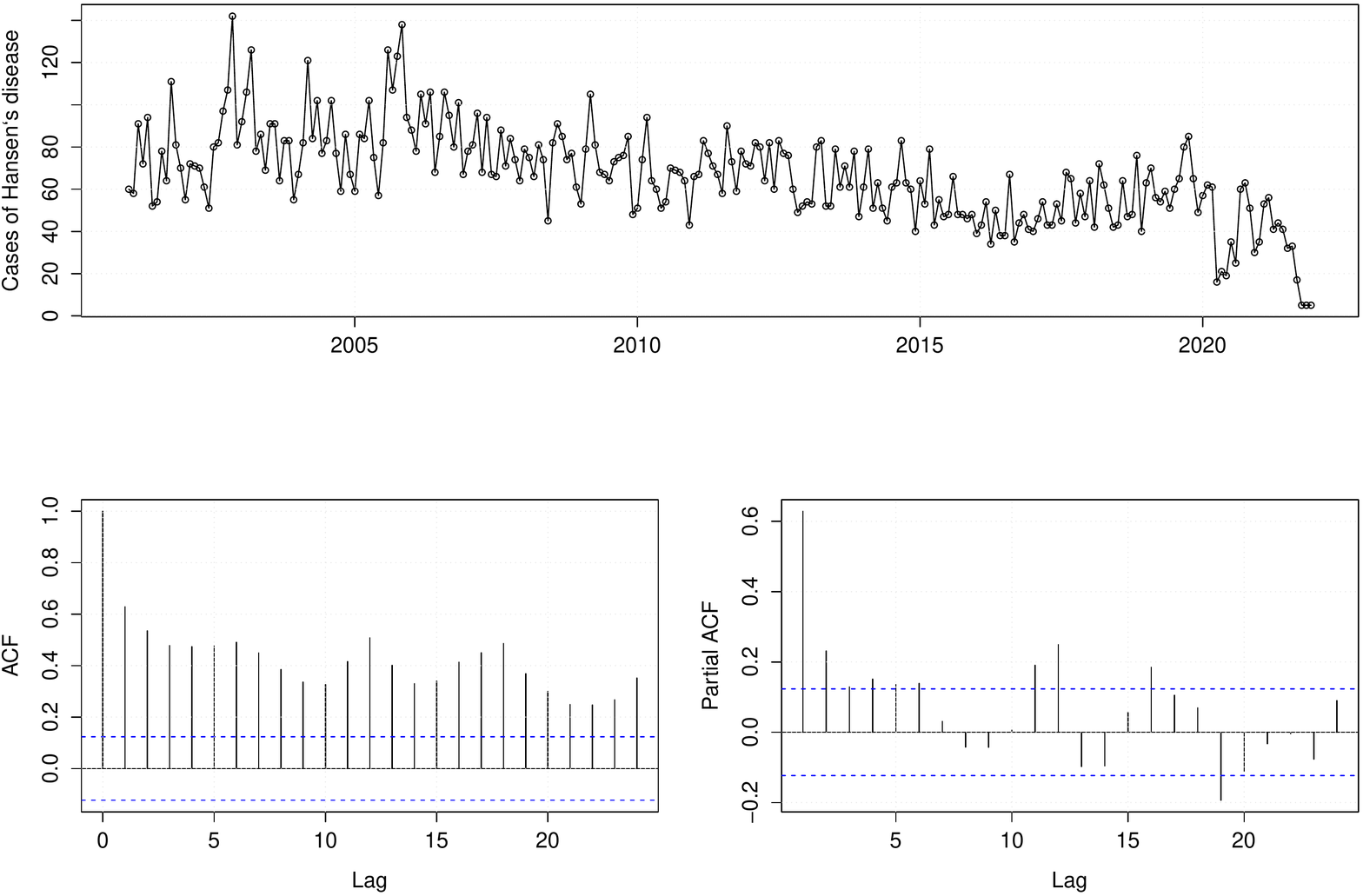}
	\caption{Hansen's disease data (top panel) and corresponding autocorrelation function (bottom left panel) and partial autocorrelation function (bottom right panel).} \label{F:hans_data}
\end{figure}

We consider the CLS estimation procedure for both approaches considered here. Table~\ref{T:results_hans} gives us the parameter estimates under the MGWI and PINAR(1) processes, standard errors obtained via bootstrap, and the SSPE values. To get the standard errors for the parameter estimates, we proceed similarly as done in the first application with a slight difference. Since here the counts are high, the geometric assumption cannot be valid. Therefore, we consider a non-stationary MGWI process with innovations following a Poisson distribution with mean $\dfrac{\mu_t(1+\mu_t)}{1+\mu_t+\alpha_t}$ in our bootstrap scheme. This ensures that the conditional mean is the same as in (\ref{condmean_nonstat}). From Table \ref{T:results_hans}, we have that the trend is significant (using, for example, a significance level at 5\%) to explain the marginal mean $\mu_t$, but not for the parameter $\alpha_t$, under the MGWI model. Furthermore, we note that the sign of the estimate of $\beta$ is negative, which is in agreement with the observed negative trend. We highlight that the parameter $\mu_t$ also appears in the autocorrelation structure under our approach, therefore the trend is also significant to explain the autocorrelation of the MGWI process. By looking at the results from the PINAR fitting, we see that the trend is significant to explain $\alpha_t$ (parameter related to the autocorrelation) but not the marginal mean $\mu_t$. Once again, we have that the model producing the smallest SSPE is the MGWI process. So, our proposed methodology is performing better than the classic PINAR model in terms of prediction. The predictive values according to both models along with the observed counts are exhibited in Figure~\ref{F:app_hans}.
 
\begin{table}[!h]
\centering
\small
\renewcommand{\arraystretch}{1.2}
\caption{\small CLS estimates, standard errors, and SSPE values of the fitted non-stationary MGWI and classic PINAR processes for the Hansen's disease data.}\label{T:results_hans}
  \begin{tabular}{ccccrrrrrr}
  \toprule
  Models && Parameters && Estimates && Stand. Errors && SSPE\\
  \midrule
  \multirow{ 4}{*}{MGWI} && $\beta_0$ && 4.3538 && 0.0671 && \multirow{ 4}{*}{58742.31} \\
  	   && $\beta_1$ && $-0.7243$ && 0.1178 &&  \\
  	   && $\gamma_0$ && 4.5297 && 0.6303 &&  \\
  	   && $\gamma_1$ && $-0.5613$ && 0.9399 &&  \\
  \midrule	   
  \multirow{ 4}{*}{PINAR} && $\beta_0$ && $-0.7668$ && 0.5332 && \multirow{ 4}{*}{59919.40} \\
  	   && $\beta_1$ && 0.7997 && 0.8876 &&  \\
  	   && $\gamma_0$ && 4.5290 && 0.0212 &&  \\
  	   && $\gamma_1$ && $-0.6883$ && 0.0433 &&  \\
  \bottomrule      
  \end{tabular}
\end{table}

\begin{figure}[!h]
	\centering		
	\includegraphics[scale=0.55]{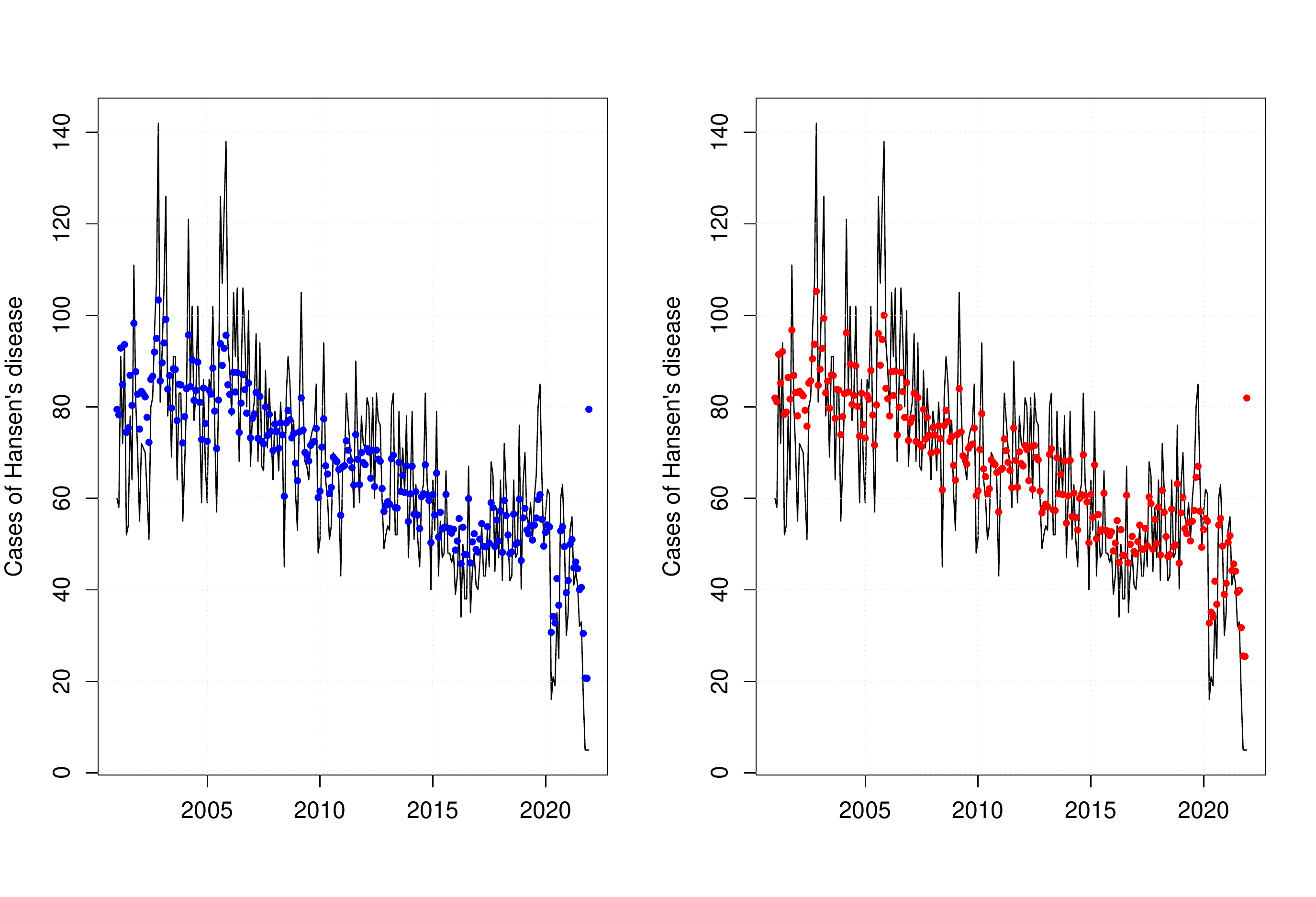}
	\vspace{-5mm}\caption{Plots of Hansen's disease data (solid line) and fitted conditional means (dots) based on the non-stationary MGWI process (to the left) and PINAR process (to the right).} \label{F:app_hans}
\end{figure}

We now conclude this data analysis by checking if the non-stationary MGWI process fits well the data. Figure~\ref{F:residuals_hans2} provides the Pearson residuals against time, its ACF plot, and the qq-plot of the residuals. By looking at this figure, we have evidence of the adequacy of the MGWI process to fit Hansen's disease data. 

\begin{figure}[!h]
	\centering		
	\includegraphics[scale=0.5]{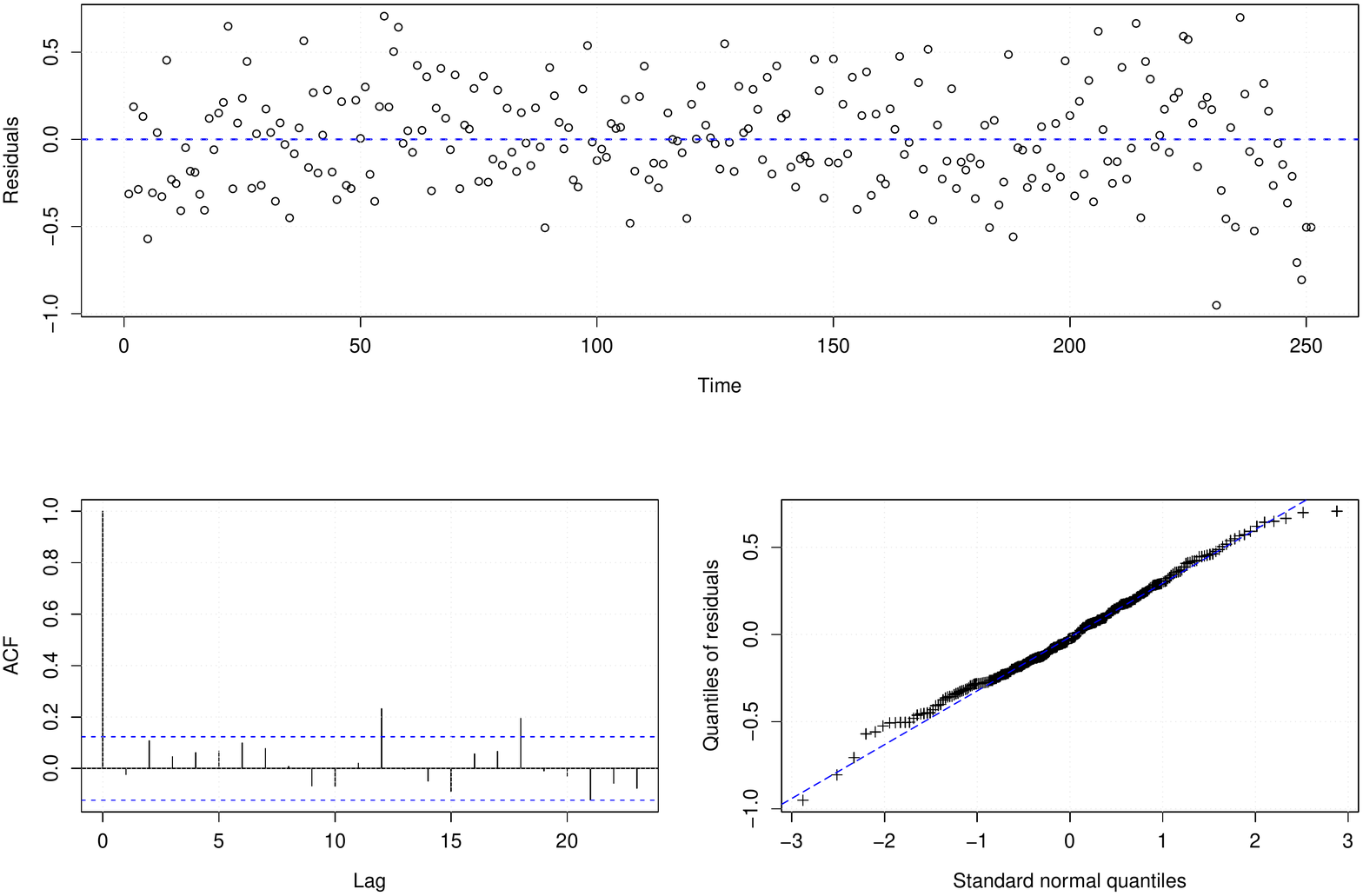}
	\caption{Pearson residuals for the MGWI process fitted to the Hansen's disease data: residuals against time (top panel), ACF (bottom left panel) and qq-plot(bottom right panel).
} \label{F:residuals_hans2}
\end{figure}

\section{Generalization}\label{sec:generalization}

In this section, we provide an extension of the geometric thinning operator and  propose a modified GWI process based on such generalization. As we will see, alternative distributions rather than geometric for the operation in (\ref{minop}) can provide flexible approaches for dealing with different features on count time series. We also discuss how to handle zero-inflation or zero-deflation with respect to the geometric model.

\begin{definition}(Zero-modified geometric (ZMG) thinning operator) Assume that $X$ is a non-negative integer-valued random variable, independent of $Z^{(\eta,\alpha)}\sim  \mathrm{ZMG}(1-\eta,\alpha)$, with $\alpha>0$ and $1-\eta\in(-1/\alpha,1)$. We define the zero-modified geometric thinning operator $(\eta,\alpha) \minop$ by
	\begin{equation}\label{zmgop}
	(\eta,\alpha) \minop X \stackrel{d}{=} \min\left(X, Z^{(\eta,\alpha)}\right).
	\end{equation}
\end{definition}

\begin{remark}
	Note that the ZMG operator given in (\ref{zmgop}) has the geometric thinning operator as a special case when $\eta=1$ since $Z^{(1,\alpha)}\sim\mathrm{Geo}(\alpha)$. Further, we stress that the parameterization of the ZMG distribution in terms of $1-\eta$ instead of $\eta$ will be convenient in what follows. Also, we will omit the dependence of $Z$ on $(\eta,\alpha)$ to simplify the notation.
\end{remark}

Based on the ZMG operator, we can define a modified GWI process $\{X_t\}_{t\in\mathbb{N}}$ (similarly as done in Section \ref{sec:gwi_process}) by
	\begin{equation}\label{mgwi2}
	X_t = (\eta,\alpha) \minop X_{t-1} + \epsilon_t,\quad t\in\mathbb N,
	\end{equation}
where $(\eta,\alpha) \minop X_{t-1}=\min\left(X_{t-1},Z_t\right)$, with $\{Z_t\}_{t\in\mathbb{N}}\stackrel{iid}{\sim}\mbox{ZMG}(1-\eta,\alpha)$, $\{\epsilon_t\}_{t\geq1}$ is a sequence of iid non-negative integer-valued random variables, called innovations, with $\epsilon_t$ independent of $X_{t-l}$ and $Z_{t-l+1}$, for all $l\geq 1$, with $X_0$ being some starting value/random variable. This is basically the same idea as before; we are just replacing the geometric assumption by the zero-modified geometric law in the thinning operation.

We now show that it is possible to construct a stationary Markov chain satisfying (\ref{mgwi2}) and having marginals ZMG-distributed; this could be seen as an alternative model to the zero-modified geometric INAR(1) process proposed by \cite{bar2015}. Furthermore, we argue that such construction is not possible under the geometric thinning operator defined in Section \ref{sec:novel_operator} (see Remark \ref{margin} below), which motivates the ZMG thinning introduced here.

Let $X{\sim}\mbox{ZMG}(1-\pi,\mu)$ with $\mu>0$ and $1-\pi\in(-1/\mu,1)$. For $z=0,1,\dots$, it holds that 
\begin{align*}
\p((\eta,\alpha) \minop X > z) &= \p(X > z)\p(Z^{(\eta,\alpha)} > z) = \pi\eta\left[\left(\frac{\mu}{1+\mu}\right)\left(\frac{\alpha}{1+\alpha}\right)\right]^{z+1}.
\end{align*}
In other words, $(\eta,\alpha) \minop X\sim\mathrm{ZMG}\left(1-\eta\pi,\frac{\mu\alpha}{1+\mu+\alpha}\right)$.
Writing $\Psi_\epsilon(s)\equiv \dfrac{\Psi_X(s)}{\Psi_{(\eta,\alpha) \minop X}(s)}$, we obtain
\begin{eqnarray}\label{conv_eps}
\Psi_\epsilon(s)&=&\left\{\dfrac{1+(1-\pi)\mu(1-s)}{1+\mu(1-s)}\right\}\Bigg/\left\{\dfrac{1+(1-\pi\eta)\frac{\mu\alpha}{1+\mu+\alpha}(1-s)}{1+\frac{\mu\alpha}{1+\mu+\alpha}(1-s)}\right\}\nonumber\\
&=&\left\{\dfrac{1+(1-\pi)\mu(1-s)}{1+(1-\pi\eta)\frac{\mu\alpha}{1+\mu+\alpha}(1-s)}\right\}\left\{\dfrac{1+\frac{\mu\alpha}{1+\mu+\alpha}(1-s)}{1+\mu(1-s)}\right\}\equiv\varphi_1(s)\varphi_2(s),
\end{eqnarray}
for all $s$ such that $|s|<1+\min(\mu^{-1},\alpha^{-1})$, where $\varphi_2(\cdot)$ denotes the pgf of a $\mathrm{ZMG}\left(\frac{\alpha}{1+\mu+\alpha},\mu\right)$ distribution. In addition to the restrictions on $\pi$ and $\eta$ above, assume that $\pi\eta<1$, $\eta\neq1$, and $\frac{1-\pi}{1-\pi\eta}\left(1+\frac{1+\mu}{\alpha}\right)<1$. Under these conditions, $\varphi_1(\cdot)$ is the pgf of a $\mathrm{ZMG}\left(\frac{1-\pi}{1-\pi\eta}\left(1+\frac{1+\mu}{\alpha}\right),(1-\pi\eta)\frac{\mu\alpha}{1+\mu+\alpha}\right)$ distribution. This implies that $\Psi_\epsilon(\cdot)$ is a proper pgf associated to a convolution between two independent ZMG random variables. Hence, we are able to introduce a MGWI process with ZMG marginals as follows. 

\begin{definition}
A stationary MGWI process $\{X_t\}_{t\in\mathbb{N}}$ with $\mbox{ZMG}(1-\pi,\mu)$ marginals (ZMG-MGWI) is defined by assuming that (\ref{mgwi2}) holds with $\{\epsilon_t\}_{t\geq1}$ being an iid sequence of random variables with pgf given by (\ref{conv_eps}), and $X_0{\sim}\mbox{ZMG}(1-\pi,\mu)$, with $\mu>0$ and $1-\pi\in(-1/\mu,1)$.
\end{definition}

\begin{remark}\label{margin}
Note that we are excluding the case $\eta=1$ (which corresponds to the geometric thinning operator) since the required inequality $\frac{1-\pi}{1-\pi\eta}\left(1+\frac{1+\mu}{\alpha}\right)<1$ does not hold in this case ($1+\frac{1+\mu}{\alpha}>1$). This shows that an MGWI process with ZMG marginals cannot be constructed based on the geometric thinning operator defined previously and therefore motivates the ZMG operator.
\end{remark}

\section*{Acknowledgments}

W. Barreto-Souza would like to acknowledge support from KAUST Research Fund. Roger Silva was partially supported by FAPEMIG, grant APQ-00774-21.


\begin{thebibliography}{100}

\bibitem[{Aleksi\'c and Risti\'c(2021)}]{aleris2021} 
\textsc{Aleksi\'c, M.S., Risti\'c, M.M.} (2021).
\newblock{A geometric minification integer-valued autoregressive model}.
\newblock{\emph{Applied Mathematical Modelling.}} \textbf{90}, 265--280.

\bibitem[{Alzaid and Al-Osh(1987)}]{alzalo1987} 
\textsc{Alzaid, A.A., Al-Osh, M.A.} (1987).
\newblock{First-order integer-valued autoregressive (INAR(1)) process}.
\newblock{\emph{Journal of Time Series Analysis.}} \textbf{8}, 261--275.

\bibitem[{Barreto-Souza(2015)}]{bar2015} 
\textsc{Barreto-Souza, W.} (2015).
\newblock{Zero-modified geometric INAR(1) process for modelling count time series with deflation or inflation of zeros}. 
\newblock{\emph{Journal of Time Series Analysis.}} \textbf{36}, 839--852.

\bibitem[{Billingsley(1995)}]{bill1995} 
\textsc{Billingsley, P.} (1995).
\newblock{Probability and Measure, 3rd edition}.
\newblock{\emph{Wiley \& Sons, New York.}}

\bibitem[{Br\"ann\"as(1995)}]{bra1995} 
\textsc{Br\"ann\"as, K.} (1995).
\newblock{Explanatory variables in the AR(1) count data model}. 
\newblock{\emph{Ume\aa\, Economic Studies.}} \textbf{381}.

\bibitem[{Dion et al.(1995)}]{dioetal1995} 
\textsc{Dion, J.P., Gauthier, G., Latour, A.} (1995).
\newblock{Branching processes with immigration and integer-valued time series}.
\newblock{\emph{Serdica Mathematical Journal.}} \textbf{21}, 123--136.

\bibitem[{Dunn and Smyth(1996)}]{dunsmy1996} 
\textsc{Dunn, P.K., Smyth, G.K.} (1996).
\newblock{Randomized quantile residuals}.
\newblock{\emph{Journal of Computational and Graphical Statistics.}} \textbf{5}, 236--244.

\bibitem[{Enciso-Mora et al.(2009)}]{encetal2009} 
\textsc{Enciso-Mora, V., Neal, P., Rao, T.S.} (2009).
\newblock{Integer valued AR processes with explanatory variables}.
\newblock{\emph{Sankhy\=a -- Series B.}} \textbf{71}, 248--263.


\bibitem[{Freeland and McCabe(2005)}]{fremcc2005} 
\textsc{Freeland, R.M., McCabe, B.P.M.} (2005).
\newblock{Asymptotic properties of CLS estimators in the Poisson AR(1) model}.
\newblock{\emph{Statistics and Probability Letters.}} \textbf{73}, 147--153.


\bibitem[{Kalamkar(1995)}]{kal1995} 
\textsc{Kalamkar, V.A.} (1995).
\newblock{Minification processes with discrete marginals}.
\newblock{\emph{Journal of Applied Probability.}} \textbf{32}, 692--706.


\bibitem[{Littlejohn(1992)}]{lit1992} 
\textsc{Littlejohn, R.P.} (1992).
\newblock{Discrete minification processes and reversibility}.
\newblock{\emph{Journal of Applied Probability.}} \textbf{29}, 82--91.


\bibitem[{Littlejohn(1996)}]{lit1996} 
\textsc{Littlejohn, R.P.} (1996).
\newblock{A reversibility relationship for two Markovian time series models with stationary geometric tailed distribution}.
\newblock{\emph{Stochastic Processes and their Applications.}} \textbf{64}, 127--133.

\bibitem[{Isp\'any et al.(2003)}]{ispetal2003} 
\textsc{Isp\'any, M., Pap, G., Van Zuijlen, M.C.A.} (2003).
\newblock{Asymptotic inference for nearly unstable INAR(1) models}.
\newblock{\emph{Journal of Applied Probability.}} \textbf{40}, 750--765.

\bibitem[{Maia et al.(2021)}]{maietal2021} 
\textsc{Maia, G.O., Barreto-Souza, W, Bastos, F.S., Ombao, H.} (2021).
\newblock{Semiparametric time series models driven by latent factor}.
\newblock{\emph{International Journal of Forecasting.}} \textbf{37}, 1463--1479.

\bibitem[{McKenzie(1988)}]{mck1988} 
\textsc{McKenzie, E.} (1988).
\newblock{Some ARMA models for dependent sequences of Poisson counts}.
\newblock{\emph{Advances in Applied Probability.}} \textbf{20}, 822--835.

\bibitem[{\texttt{R} Core Team (2021)}]{r2021} 
\textsc{\texttt{R} Core Team} (2021).
\newblock{\texttt{R}: A language and environment for statistical computing}.
\newblock{R Foundation for Statistical Computing, Vienna, Austria.}
\newblock{\small URL \url{https://www.R-project.org/}.}

\bibitem[{Rahimov(2008)}]{rah2008} 
\textsc{Rahimov, I.} (2008).
\newblock{Asymptotic distribution of the CLSE in a critical process with immigration}.
\newblock{\emph{Stochastic Processes and their Applications.}} \textbf{118}, 1892--1908.

\bibitem[{Scotto et al.(2016)}]{scoetal2016} 
\textsc{Scotto, M.G., Wei\ss, C.H., M\"oller, T.A., Gouveia, S.} (2016).
\newblock{The max-INAR(1) model for count processes}.
\newblock{\emph{Test.}} \textbf{27}, 850--870.


\bibitem[{Silva and Barreto-Souza(2019)}]{silbar2019} 
\textsc{Silva, R.B., Barreto-Souza, W.} (2019).
\newblock{Flexible and robust mixed Poisson INGARCH models}.
\newblock{\emph{Journal of Time Series Analysis.}} \textbf{40}, 788--814.




\bibitem[{Steutel and van Harn(1979)}]{steandvan1979} 
\textsc{Steutel, F.W., van Harn, K.} (1979).
\newblock{Discrete analogues of self-decomposability and stability}.
\newblock{\emph{Annals of Probability.}} \textbf{7}, 893--899.

\bibitem[{Wang(2020)}]{wan2020} 
\textsc{Wang, X.} (2020).
\newblock{Variable selection for first-order Poisson integer-valued autoregressive model with covariables}.
\newblock{\emph{Australian and New Zealand Journal of Statistics.}} \textbf{62}, 278--295.


\bibitem[{Wei and Winnicki(1990)}]{weiwin1990} 
\textsc{Wei, C.Z., Winnicki,  J.} (1990).
\newblock{Estimation of the mean in the branching process with immigration}.
\newblock{\emph{Annals of Statistics.}} \textbf{18}, 1757--1773.


\bibitem[{Zhu(2011)}]{zhu2011} 
\textsc{Zhu, F.} (2011).
\newblock{A negative binomial integer-valued GARCH model}.
\newblock{\emph{Journal of Time Series Analysis.}} \textbf{32}, 54--67.

\end{thebibliography}
\end{document}